\documentclass[times]{nmearxiv} %
\usepackage{moreverb}
\usepackage{amsthm}
\usepackage[colorlinks,bookmarksopen,bookmarksnumbered,citecolor=red,urlcolor=red]{hyperref}
\usepackage{todonotes}

\newcommand\BibTeX{{\rmfamily B\kern-.05em \textsc{i\kern-.025em b}\kern-.08emT\kern-.1667em\lower.7ex\hbox{E}\kern-.125emX}}

\usepackage{epsfig}
\usepackage{algorithmic, algorithm}

\def\norm#1{\|#1\|}

\def\real{\mathbb{R}}
\def\gauss{\mathcal{N}}
\def\mse{{\rm MSE}}
\def\var{{\rm Var}}
\def\bias{{\rm Bias}}
\def\ess{{\rm ESS}}
\def\iact{{\rm IACT}}
\def\corr{{\rm corr}}

\def\vecx{{\bf x}}
\def\vecy{{\bf y}}
\def\vecd{{\bf d}}
\def\data{{\bf d}_{\rm obs}}
\def\vecu{{\bf u}}
\def\vecq{{\bf q}}

\def\vece{{\bf e}}
\def\vecf{{\bf f}}

\def\vect{{\bf t}}

\def\matA{{\rm \bf A}}
\def\matC{{\rm \bf C}}
\def\matV{{\rm \bf V}}

\def\dhell{d_{\rm Hell}}

\theoremstyle{plain}
\newtheorem{theorem}{Theorem}
\newtheorem{assumption}[theorem]{Assumption}
\newtheorem{lemma}[theorem]{Lemma}
\newtheorem{definition}[theorem]{Definition}
\newtheorem{proposition}[theorem]{Proposition}

\usepackage{array}
\newcolumntype{L}[1]{>{\raggedright\let\newline\\\arraybackslash\hspace{0pt}}m{#1}}
\newcolumntype{C}[1]{>{\centering\let\newline\\\arraybackslash\hspace{0pt}}m{#1}}
\newcolumntype{R}[1]{>{\raggedleft\let\newline\\\arraybackslash\hspace{0pt}}m{#1}}

\runningheads{T.\ Cui, Y.\ Marzouk, and K.\ Willcox}{Data-Driven Model Reduction for the Bayesian Solution of Inverse Problems}
\title{Data-Driven Model Reduction for the Bayesian Solution of Inverse Problems}

\author{Tiangang Cui\affil{1}, Youssef M.\ Marzouk\affil{1}, Karen E.\ Willcox\affil{1}\corrauth}
\address{\affilnum{1}Department of Aeronautics and Astronautics, Massachusetts Institute of Technology, Cambridge, MA 02139, USA}
\corraddr{K.\ E.\ Willcox, Department of Aeronautics and Astronautics, Massachusetts Institute of Technology, Cambridge, MA 02139, USA. Emails: {\tt tcui@mit.edu, ymarz@mit.edu, kwillcox@mit.edu}}

\cgsn{DOE Office of Advanced Scientific Computing Research (ASCR)}{DE-FG02-08ER2585 and DE-SC0009297}

\begin{document}

\begin{abstract}
One of the major challenges in the Bayesian solution of inverse problems governed by partial differential equations (PDEs) is the computational cost of repeatedly evaluating numerical PDE models, as required by Markov chain Monte Carlo (MCMC) methods for posterior sampling.
This paper proposes a data-driven projection-based model reduction technique to reduce this computational cost.
The proposed technique has two distinctive features. First, the model reduction strategy is tailored to inverse problems: the snapshots used to construct the reduced-order model are computed adaptively from the posterior distribution. Posterior exploration and model reduction are thus pursued simultaneously.
Second, to avoid repeated evaluations of the full-scale numerical model as in a standard MCMC method, we couple the full-scale model and the reduced-order model together in the MCMC algorithm. This maintains accurate inference while reducing its overall computational cost.
In numerical experiments considering steady-state flow in a porous medium, the data-driven reduced-order model achieves better accuracy than a reduced-order model constructed using the classical approach. It also improves posterior sampling efficiency by several orders of magnitude compared to a standard MCMC method.
\end{abstract}

\keywords{model reduction; inverse problem; adaptive Markov chain Monte Carlo; approximate Bayesian inference}

\maketitle

\section{Introduction and motivation}

An important and challenging task in computational modeling is the solution of inverse problems, which convert noisy and indirect observational data into useful characterizations of the unknown parameters of a numerical model. In this process, statistical methods---Bayesian methods in particular---play a fundamental role in modeling various information sources and quantifying the uncertainty of the model parameters \cite{Kaipio_2005,Tarantola_2004}. In the Bayesian framework, the unknown parameters are modeled as random variables and hence can be characterized by their posterior distribution. Markov chain Monte Carlo (MCMC) methods \cite{Liu_2001} provide a powerful and flexible approach for sampling from posterior distributions.
The Bayesian framework has been applied to inverse problems in various fields, for example, geothermal reservoir modeling \cite{CFO_2011}, groundwater modeling \cite{HLH_2003}, ocean dynamics \cite{Mckeague_2005}, remote sensing \cite{Haario_2004}, and seismic inversion \cite{Martin_2012}.

To generate a sufficient number of samples from the posterior distribution, MCMC methods require sequential evaluations of the posterior probability density at many different points in the parameter space.
Each evaluation of the posterior density involves a solution of the forward model used to define the likelihood function, which typically is a computationally intensive undertaking (e.g., the solution of a system of PDEs). 
In this many-query situation, one way to address the computational burden of evaluating the forward model is to replace it with a computationally efficient surrogate.
Surrogate models have been applied to inverse problems in several settings; for example, \cite{Marzouk_2007, Marzouk_2009} employed generalized polynomial chaos expansions, \cite{Bayarri_2009} employed Gaussian process models, and \cite{Galbally_2008,Lipponen_2013,Lieberman_2010,Wang_2005} used projection-based reduced-order models.
In this work, we also focus on projection-based reduced-order models (although the data-driven strategy underlying our approach should be applicable to other types of surrogate models).
A projection-based reduced-order model reduces the computational complexity of the original or ``full'' forward model by solving a projection of the full model onto a reduced subspace.
For the model reduction methods we consider, the construction of the reduced subspace requires evaluating the full model at representative samples drawn from the parameter space.
The solutions of the full model at these samples are referred to as snapshots \cite{Sirovich_1987}. Their span defines the reduced subspace, represented via an orthogonal basis.

The quality of the reduced-order model relies crucially on the choice of the samples for computing the snapshots.
To construct reduced-order models targeting the Bayesian solution of the inverse problem, we employ existing projection-based model reduction techniques. Our innovation is in a data-driven approach that adaptively selects samples from the posterior distribution for the snapshot evaluations.
This approach has two distinctive features:
\begin{enumerate}
\item
We integrate the reduced-order model construction process into an adaptive MCMC algorithm that simultaneously samples the posterior and selects posterior samples for computing the snapshots.
During the sampling process, the numerical accuracy of the reduced-order model is adaptively improved.
\item The approximate posterior distribution defined by the reduced-order model is used to increase the efficiency of MCMC sampling.
We either couple the reduced-order model and the full model together to accelerate the sampling of the full posterior distribution\footnote{The term ``full posterior'' refers to the posterior distribution induced by the original or full forward model.}, or directly explore the approximate posterior distribution induced by the reduced-order model.
In the latter case, sampling the approximate distribution yields a biased Monte Carlo estimator, but the bias can be controlled using error indicators or estimators.
\end{enumerate}

\begin{figure}[h!]
\centering
\includegraphics[width=\textwidth]{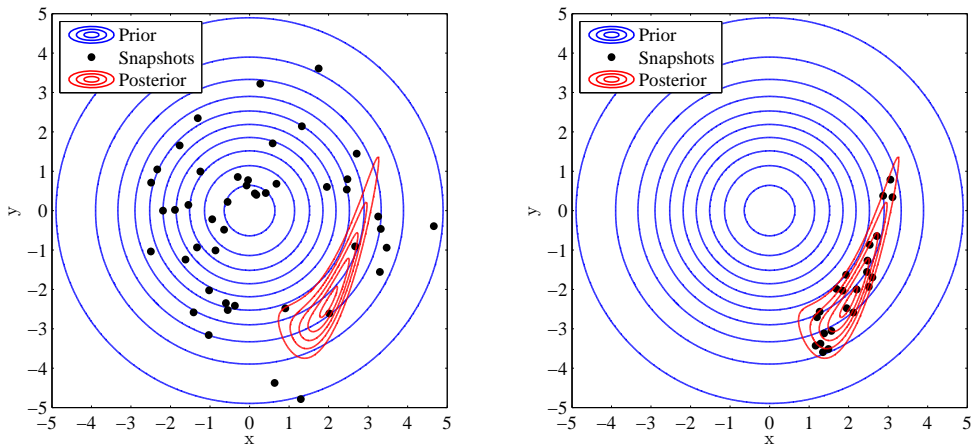}
\caption{A two-dimensional example demonstrating data-driven model reduction. The prior distribution and posterior distribution are represented by the blue contours and red contours, respectively. The black dots represent samples used for computing the snapshots in the reduced-order model construction. Left: sampling using the classical approach (from the prior). Right: our data-driven approach.}
\label{fig:prior_post}
\end{figure}
Compared to the classical offline approaches that build the reduced-order model before using it in the many-query situation, the motivation for collecting snapshots during posterior exploration is to build a reduced-order model that focuses on a more concentrated region in the parameter space.
Because the solution of the inverse problem is unknown until the data are available, reduced-order models built offline have to retain a level of numerical accuracy over a rather large region in the parameter space, which covers the support of the posterior distributions for all the possible observed data sets.
For example, samples used for computing the snapshots are typically drawn from the prior distribution \cite{Lipponen_2013}.
In comparison, our data-driven approach focuses only on the posterior distribution resulting from a particular data set.
As the observed data necessarily increase the information divergence of the prior distribution from the posterior distribution \cite{Cover_2012}, the support of the posterior distribution is 
more compact than that of the prior distribution.

Figure \ref{fig:prior_post} illustrates a simple two-dimensional inference problem, where the prior distribution and the posterior distribution are represented by the blue and red contours, respectively.
The left plot of Figure~\ref{fig:prior_post} shows 50 randomly drawn prior samples for computing the snapshots, each of which has a low probability of being in the support of the posterior. 
In comparison, the samples selected by our data-driven approach, as shown in the right plot of Figure \ref{fig:prior_post}, are scattered within the region of high posterior probability.

By retaining numerical accuracy only in a more concentrated region, the data-driven reduced-order model requires a basis of lower dimension to achieve the same level of accuracy compared to the reduced-order model built offline.
For the same reason, the data-driven model reduction technique can potentially have better scalability with parameter dimension than the offline approach.

We note that goal-oriented model reduction approaches have been developed in the context of PDE-constrained optimization \cite{Arian_2000, Ravindran_2002, KV_2008, Carlberg_2008}, in which the reduced-order model is simultaneously constructed during the optimization process. 
In these methods, the snapshots for constructing the reduced basis are only evaluated at points in the parameter space that are close to the trajectory of the optimization algorithm. 

The remainder of this paper is organized as follows. In Section 2, we outline the Bayesian framework for solving inverse problems and discuss the sampling efficiency and Monte Carlo error of MCMC. In Section 3, we introduce the data-driven model reduction approach, and construct the data-driven reduced-order model within an adaptive delayed acceptance algorithm to speed up MCMC sampling. We also provide results on the ergodicity of the algorithm. Section 4 analyzes some properties of the data-driven reduced-order model. In Section 5, we discuss a modified framework that adaptively constructs the reduced-order model and simultaneously explores the induced approximate posterior distribution. We also provide an analysis of the mean square error of the resulting Monte Carlo estimator. In Section 6, we demonstrate and discuss various aspects of our methods through numerical examples. Section 7 offers concluding remarks.

\section{Sample-based inference for inverse problems}
The first part of this section provides a brief overview of the Bayesian framework for inverse problems.
Further details can be found in \cite{Kaipio_2005, Stuart_2010, Tarantola_2004}.
The second part of this section discusses the efficiency of MCMC sampling for computationally intensive inverse problems.

\subsection{Posterior formulation and sampling}
Given a physical system, let $\vecx \in \mathbb{X} \subset \real^{N_{\rm p}}$ denote the $N_{\rm p}$-dimensional unknown parameter, and $\data \in \mathbb{D} \subset \real^{N_{\rm d}}$ denote the $N_{\rm d}$-dimensional observed data.
The forward model $\vecd = F(\vecx)$ maps a given parameter $\vecx$ to the observable model outputs $\vecd$.

In a Bayesian setting the first task is to construct the prior models and the likelihood function as probability distributions.
The prior density is a stochastic model representing knowledge of the unknown $\vecx$ {\it before}  any measurements and is denoted by $\pi_0(\vecx)$.
The likelihood function specifies the probability density of the observation $\data$ for a given set of parameters $\vecx$, denoted by $L(\data|\vecx)$. 
We assume that the data and the model parameters follow the stochastic relationship
\begin{equation}
\data = F(\vecx) + \vece,
\end{equation}
where the random vector $\vece$ captures the measurement noise and other uncertainties in the observation-model relationship, including model errors. Without additional knowledge of the measurement process and model errors, $\vece$ is modeled as a zero mean Gaussian, $\vece \sim \gauss(0, \Sigma_e)$, where $\Sigma_e$ is the covariance.
Let
\begin{equation}
\Phi(\vecx) = \frac{1}{2}  \| {\Sigma_e}^{-\frac12} \left( F(\vecx) - \data \right) \|^2
\label{eq:misfit}
\end{equation}
denote the data-misfit function.
The resulting likelihood function is proportional to
\(
\exp \left( -\Phi(\vecx) \right ).
\)
By Bayes' formula, the posterior probability density is
\begin{equation}
\label{eq:post}
\pi(\vecx|\data) = \frac{1}{Z} \exp \left( -\Phi(\vecx) \right )  \pi_0(\vecx),
\end{equation}
where
\begin{equation}
Z = \int_\mathbb{X} \exp \left( -\Phi(\vecx) \right ) \pi_0(\vecx) d\vecx.
\label{eq:normal}
\end{equation}
is the normalizing constant.

In the general framework set up by \cite{Grenander_1994}, we can explore the posterior distribution given by (\ref{eq:post}) using MCMC methods such as the Metropolis-Hastings algorithm \cite{Hastings, Metropolis}. The Metropolis-Hastings algorithm uses a proposal distribution and an acceptance/rejection step to construct the transition kernel of a Markov chain that converges to the desired target distribution.

\subsection{Sampling efficiency and Monte Carlo error}
\label{sec:monte_carlo}
Once we draw a sufficient number of samples from the posterior distribution,
the expectations of functions over the posterior distribution can be estimated by Monte Carlo integration.
Suppose we wish to estimate the expectation of a function $h(x)$ over the posterior distribution
\begin{equation}
\label{eq:exp}
I(h) = \int_\mathbb{X} h(\vecx) \pi(\vecx|\data) d\vecx,
\end{equation}
by $N$ posterior samples, $\vecx^{(i)} \sim \pi(\vecx|\data), i = 1,\ldots,N$.
The resulting Monte Carlo estimator of $I(h)$ is
\begin{equation}
\label{eq:mc}
\widehat{I(h)} = \frac{1}{N}\sum_{i=1}^{N} h(\vecx^{(i)}),
\end{equation}
which is an unbiased estimator with the mean square error %
\begin{equation}
\mse \left(\widehat{I(h)}\right) = \var \left(\widehat{I(h)}\right) = \frac{\var(h)}{\ess(h)}.
\end{equation}
Because the samples drawn by an MCMC algorithm are correlated, the variance of $\widehat{I(h)}$ is dependent on the effective sample size %
\begin{equation}
\label{eq:ess}
\ess(h) = \frac{N}{2 \times \iact(h)} ,
\end{equation}
where
\[
\iact(h) = \frac{1}{2} + \sum_{j=1}^{\infty} \corr \left( h\left(\vecx^{(1)}\right), h\left(\vecx^{(j+1)}\right)\right)
\]
is the integrated autocorrelation of $h(x)$; see \cite[Chapter 5.8]{Liu_2001} for a detailed derivation and further references.

In practice, we wish to improve the {\bf computational efficiency} of the MCMC algorithm, which is defined by the effective sample size for a given budget of CPU time. There are two ways to achieve this:
\begin{enumerate}
\item Improve the {\bf statistical efficiency}.
To increase the effective sample size for a given number of MCMC iterations, one can reduce the sample correlation by designing proposals that traverse the parameter space more efficiently. For example, adaptive MCMC methods \cite{Haario_2001, Roberts_2007} use adaptation to learn the covariance structure of the posterior distribution online, and use this information to design proposal distributions. Algorithms such as stochastic Newton \cite{Martin_2012} and Riemannian manifold MCMC \cite{Girolami_2011} use local derivative information to construct efficient proposals that adapt to the local geometry of the posterior distribution.

\item Increase the number of MCMC steps for a given amount of CPU time. Because simulating the forward model in the data-misfit function (\ref{eq:misfit}) is CPU-intensive, computing the posterior density at every iteration is the major bottleneck of the standard MCMC algorithm. By using a fast approximation to the forward model, the amount of computing time can be reduced.
This effort complements the use of MCMC algorithms that require local derivatives, since fast approximations of the forward model also enable fast evaluations of its gradient and even higher-order derivatives. \footnote{Using approximate derivatives does not impact the ergodicity of the MCMC algorithm, since the bias caused by the approximate derivatives is corrected by the acceptance/rejection step in the Metropolis-Hastings algorithm.}
\end{enumerate}

One important remark is that by drawing samples directly from an approximate posterior distribution (i.e., the one induced by our approximation of the forward model), the resulting Monte Carlo estimator (\ref{eq:mc}) is {\bf biased}.
Algorithms such as surrogate transition \cite{Liu_1998} or delayed acceptance \cite{CF_2005} MCMC can be used to couple the posterior distribution and its approximation together to ensure the samples are drawn from the full posterior distribution.
On the other hand, if the accuracy of the posterior approximation can be tightly controlled, then some bias may be acceptable if it enables significant variance reduction, and thus an overall reduction in the mean squared error of an estimate of a posterior expectation, for a given computational effort.
We explore both options in this work.

\section{Data-driven reduced-order model and full target algorithm}
\label{sec:online}

This section introduces our data-driven model reduction approach, and the adaptive sampling framework for simultaneously constructing the reduced-order model and exploring the posterior distribution.

\subsection{Posterior approximation}
Suppose the system of interest can be described by a system of nonlinear steady partial differential equations (PDEs), with a finite element or finite difference discretization that results in a discretized system in the form of
\begin{equation}
\label{eq:forward_a}
\matA(\vecx) \vecu + \vecf(\vecx, \vecu) + \vecq(\vecx) = 0.
\end{equation}
In Equation (\ref{eq:forward_a}), $\vecu \in \real^{n}$ represents the discretized state of the system, $n$ is the dimension of the system (the number of unknowns), $\matA(\vecx) \in \real^{n \times n}$ is a discretized linear operator, $\vecf(\vecx, \vecu) \in \real^{n}$ represents the discretized nonlinear terms of the governing PDE, and $\vecq(\vecx)$ denotes the forcing terms.
All of $\matA$, $\vecf$, and $\vecq$ can be parameterized by the unknown parameter $\vecx$.
An observation operator $\matC$ maps the state of the system to the observable model outputs, i.e.,
\begin{equation}
\label{eq:forward_b}
\vecd = \matC (\vecu, \vecx ).
\end{equation}
Equations (\ref{eq:forward_a}) and (\ref{eq:forward_b}) define the \textit{forward model}, $\vecd = F(\vecx)$, that maps a realization of the unknown parameter $\vecx$ to observable model outputs \vecd.

Our data-driven model reduction approach selects a set of $m$ posterior samples $ \{ \vecx^{(0)}, \ldots, \vecx^{(m)} \} $ to compute the snapshots $ \{ \vecu^{(0)}, \ldots, \vecu^{(m)} \} $ by solving Equation (\ref{eq:forward_a}) for each sample $\vecx^{(j)}$.
By constructing the reduced basis \( \matV_{m} = {\rm span}\{ \vecu^{(0)}, \ldots, \vecu^{(m)}\}\), the state $\vecu$ can be approximated by a linear combination of the reduced basis vectors, i.e., \( \vecu \approx \matV_{m}\vecu_m \).
Then Equation (\ref{eq:forward_a}) can be approximated by applying the Galerkin projection:
\begin{equation}
\label{eq:reduced_a}
\underbrace{{\matV_{m}}^T \matA(\vecx) {\matV_{m}}}_{\matA_m(\vecx)} \vecu_m + {\matV_{m}}^T \vecf(\vecx, {\matV_{m}} \vecu_m) + \underbrace{{\matV_{m}}^T\vecq(\vecx)}_{\vecq_m(\vecx)} = 0,
\end{equation}
and the associated model outputs are
\begin{equation}
\label{eq:reduced_b}
\vecd_m = \matC (\matV_m \vecu_m, \vecx ).
\end{equation}
If ${\rm dim}(\matV_m) \ll n$, the dimension of the unknown state in (\ref{eq:reduced_a}) and (\ref{eq:reduced_b}) is greatly reduced compared to that of the original full system (\ref{eq:forward_a}) and (\ref{eq:forward_b}). Thus (\ref{eq:reduced_a}) and (\ref{eq:reduced_b}) define a reduced-order model $\vecd_m = F_{m} (\vecx)$ that maps a realization of the unknown parameter $\vecx$ to an approximation of the observable model outputs $\vecd_m$.

Care must be taken to ensure efficient solution of the reduced-order model, since in the presence of general parametric dependence, the equations (\ref{eq:reduced_a}) and (\ref{eq:reduced_b}) have low state dimension but are not necessarily fast to solve.
This is because for each new parameter $\vecx$, solution of the reduced-order model requires evaluating the full scale system matrices or residual, projecting those matrices or residual onto the reduced basis, and then solving the resulting low dimensional system.
Since many elements of these computations depend on $n$, the dimension of the original system, this process is typically computationally expensive (unless there is special structure to be exploited, such as affine parametric dependence).
To reduce the computational cost of this process, methods such as the missing point estimation \cite{Astrid_2008}, the empirical interpolation \cite{BMNP_2004} and its discrete variant \cite{Chat_2010} approximate the nonlinear term in the reduced-order model by selective spatial sampling.

We note that for systems that exhibit a wide range of behaviors, recently developed localization approaches such as \cite{HDO_2011, AZF_2012, ES_2012, PBWB_2014} adaptively construct multiple local reduced bases, each tailored to a particular system behavior that is associated with a subdomain of the state space or the parameter space.
Our use of adaptation is different, focusing on the adaptive selection of posterior samples for evaluating the snapshots.
In this paper we consider only the case of a global reduced basis; however, future work could combine our adaptive posterior sampling approach with a localization strategy to build multiple local reduced bases, each adapted to the local structure of the posterior.

The data-misfit function (\ref{eq:misfit}) can be approximated by replacing the forward model $F(\cdot)$ with the reduced-order model $F_{m}(\cdot)$, which gives
\begin{equation}
\Phi_{m}(\vecx) = \frac{1}{2}  \| {\Sigma_e}^{-\frac12} ( F_{m}\left(\vecx) - \data \right) \|^2.
\label{eq:misfit_a}
\end{equation}
Then the resulting approximate posterior distribution has the form
\begin{equation}
\pi_{m}(\vecx|\data) = \frac{1}{Z_{m}} \exp \left( - \Phi_{m}(\vecx) \right ) \pi_0(\vecx),
\label{eq:post_r}
\end{equation}
where $Z_{m} = \int_\mathbb{X} \exp \left( -\Phi_m(\vecx) \right ) \pi_0(\vecx) d\vecx$ is the normalizing constant.

Our data-driven approach adaptively selects the sample for evaluating the next snapshot such that the scaled error of the reduced-order model outputs,
\begin{equation}
\vect_m(\vecx) = {\Sigma_e}^{-\frac12}\left( F(\vecx) - F_m(\vecx) \right),
\label{eq:scaled_error}
\end{equation}
for the current reduced basis $\matV_m$, is above a user-specified threshold.
The whitening transformation ${\Sigma_e}^{-\frac12}$ computes the relative error of the reduced-order model compared with noise level of the observed data, which has standard deviation 1 after the transformation.
In case that we wish to bypass the forward model evaluations, so that the error $\vect_m(\vecx)$ cannot be computed directly, an error indicator or an error estimator,  $\widehat{\vect}_m(\vecx)\approx \vect_m(\vecx)$, can be used.
In this work, we use the dual weighted residual technique \cite{Meyer_2003} to compute an error indicator.

\subsection{Full target algorithm}
\label{sec:full_algo}

To achieve simultaneous model reduction and posterior exploration, we employ the adaptive delayed acceptance algorithm of \cite{Cui_2010}.
Suppose we have a reduced-order model constructed from an initial reduced basis.
At each iteration of the MCMC sampling, we first sample the approximate posterior distribution based on the reduced-order model for a certain number of steps using a standard Metropolis-Hastings algorithm.
This first-stage subchain simulation should have sufficient number of steps, so that its initial state and last state are uncorrelated. 
Then the last state of the subchain is used as a proposal candidate in the second-stage of the algorithm, and the acceptance probability for this candidate is computed based on the ratio of the full posterior density value to the approximate posterior density value. 
This delayed acceptance algorithm employs the first-stage subchain to decrease the sample correlation by sampling the computationally fast approximate posterior, then uses the acceptance/rejection in the second-stage to ensure that the algorithm correctly samples the full posterior distribution.

The statistical efficiency of the delayed acceptance algorithm relies on the accuracy of the approximate posterior.
An approximate posterior induced by an inaccurate reduced-order model can potentially result in a higher second-stage rejection rate, which decreases the statistical efficiency of the delayed acceptance algorithm. 
This is because duplicated MCMC samples generated by rejections increase the sample correlation (we refer to \cite{CF_2005} for a formal justification).
To maintain statistical efficiency, we aim to construct a sufficiently accurate reduced-order model, so that the resulting second-stage acceptance probability is close to $1$.
To achieve this, we introduce adaptive reduced basis enrichment into the delayed acceptance algorithm.

After each full posterior density evaluation, the state of the associated forward model evaluation can be used as a potential new snapshot. %
We compute the scaled error (\ref{eq:scaled_error}) of the reduced-order model outputs at each new posterior sample, and the reduced basis is updated with the new snapshot when the error exceeds a user-given threshold $\epsilon$.
By choosing an appropriate threshold $\epsilon$, we can control the maximum allowable amount of error of the reduced-order model.

The resulting reduced-order model is data-driven, since it uses the information provided by the observed data (in the form of the posterior distribution) to select samples for computing the snapshots.
It is also an online model reduction approach, since the reduced-order model is built concurrently with posterior sampling.

\subsubsection{The algorithm.}

Since the adaptive sampling procedure samples from the full posterior distribution, hereafter we refer to it as the ``full target algorithm.''
Details of the full target algorithm are given in Algorithm \ref{algo:full}.%
\begin{algorithm}[ht]
\caption{Full target algorithm}
\label{algo:full}
\begin{algorithmic}[1]
  \REQUIRE Given the subchain length $L$, the maximum allowable reduced basis dimension $M$, and the error threshold $\epsilon$. At step $n$, given state $X_n = \vecx$, a proposal distribution $q(\vecx, \cdot)$, and a reduced-order model $F_{m}(\cdot)$ defined by the reduced basis $\matV_{m}$, one step of the algorithm is:
  \STATE Set $Y_1 = X_n$, $\vecy = \vecx$, and $i = 1$
  \WHILE{ $i \leq L$ }
  \STATE Propose a candidate $\vecy^{\prime} \sim q(\vecy, \cdot)$, then evaluate the acceptance probability
  \[
  \alpha(\vecy, \vecy^{\prime}) = 1 \wedge \frac{\pi_{m}(\vecy^{\prime}|\data)}{\pi_{m}(\vecy|\data)} \frac{q(\vecy^{\prime}, \vecy)}{q(\vecy, \vecy^{\prime})}
  \]
  	 \IF{${\rm Uniform}(0, 1] < \alpha(\vecy, \vecy^{\prime})$}
  	   \STATE Accept $\vecy^{\prime}$ by setting $Y_{i+1} = \vecy^{\prime}$
  	 \ELSE
  	   \STATE Reject $\vecy^{\prime}$ by setting $Y_{i+1} = \vecy$
  	 \ENDIF
  \STATE $i = i+1$	
  \IF{{\bf finite adaptation criterion} not satisfied \AND $\norm{\widehat{\vect}_{m}(X_{n+1})}_{\infty} \geq \epsilon$}
  \STATE \textbf{break}
  \ENDIF
  \ENDWHILE
  \STATE Set $\vecx^{\prime} = Y_{i}$, then evaluate the acceptance probability using the full posterior
  \[
  \beta(\vecx, \vecx^{\prime}) = 1 \wedge \frac{\pi(\vecx^{\prime}|\data)}{\pi(\vecx|\data)} \frac{\pi_{m}(\vecx|\data)}{\pi_{m}(\vecx^{\prime}|\data)}
  \]
  	 \IF{${\rm Uniform}(0, 1] < \beta(\vecx, \vecx^{\prime})$}
  	   \STATE Accept $\vecx^{\prime}$ by setting $X_{n+1} = \vecx^{\prime}$
  	 \ELSE
  	   \STATE Reject $\vecx^{\prime}$ by setting $X_{n+1} = \vecx$
  	 \ENDIF
  \IF{{\bf finite adaptation criterion} not satisfied \AND $m < M$ \AND $\norm{\vect_{m}(X_{n+1})}_{\infty} \geq \epsilon$}
      \STATE Update the reduced basis $\matV_{m}$ to $\matV_{m+1}$ using the new full model evaluation at $\vecx^{\prime}$ by a Gram-Schmidt process
  \ENDIF
\end{algorithmic}
\end{algorithm}

Lines 1--13 of Algorithm \ref{algo:full} simulate a Markov chain with invariant distribution $\pi_{m}(\vecy|\data)$ for $L$ steps.
In Lines 14--19 of Algorithm \ref{algo:full}, to ensure that the algorithm samples from the full posterior distribution, we post-process the last state of the subchain using a second-stage acceptance/rejection based on the full posterior distribution.

The second-stage acceptance probability $\beta$ is controlled by the accuracy of the reduced-order model and the subchain length $L$.
Ideally, we want to have large $L$ to give uncorrelated samples.
However, if the second-stage acceptance probability $\beta$ is low, the effort spent on simulating the subchain will be wasted because the proposal is more likely to be rejected in the second step.
To avoid this situation, we dynamically monitor the accuracy of the reduced-order model by evaluating its
error indicator at each state of the subchain (Lines 10--12).
The subchain is terminated if the L-infinity norm of the scaled error indicator exceeds a threshold $\epsilon$.

Lines 20--22 of Algorithm \ref{algo:full} describe the adaptation of the reduced basis, which is controlled by the scaled error (\ref{eq:scaled_error}) of the reduced-order model and a given threshold $\epsilon$.
Three criteria must be satisfied for the reduced basis to be updated: ({\romannumeral 1}) the scaled error at $X_{n+1}$ must exceed the threshold $\epsilon$; ({\romannumeral 2}) the dimensionality of the reduced basis must not exceed the maximum allowable dimensionality $M$; and ({\romannumeral 3}) the finite adaptation criterion (Definition~\ref{fin_a}) should not yet be satisfied. The finite adaptation criterion is defined precisely as follows. 

\begin{definition}
\label{fin_a}
{\rm Finite adaptation criterion}. The average number of MCMC steps used for each reduced basis enrichment exceeds $N_{\rm max} = \frac{1}{c \epsilon}$, for some user-specified $c > 0$, where $\epsilon$ is the error threshold used in Algorithm~\ref{algo:full}.
\end{definition}
The finite adaptation criterion is a threshold for how \textit{infrequent} updates to the reduced-order model should become before model adaptation is stopped entirely. When more and more MCMC steps are used between updates, the reduced-order model has satisfied the accuracy threshold $\epsilon$ over larger and larger regions of the posterior. The threshold of at least $N_{\mathrm{max}}$ steps between updates can thus be understood as a measure of ``sufficient coverage'' by the reduced-order model. Further discussion and justification of the finite adaptation criterion is deferred to Section 4. 

Algorithm \ref{algo:full} uses the adaptation in Lines 20--22 to drive the online construction of the reduced-order model. Once the adaptation stops, the algorithm runs as a standard delayed acceptance algorithm, with an $L$-step subchain in the first stage.

To prevent the situation where the reduced basis dimension is large so that the reduced-order model becomes computationally inefficient, we limit the reduced basis dimension to be less than a user-given threshold $M$. 
When the reduced basis dimension reaches $M$ but the finite adaptation criterion is not satisfied, one strategy is to stop the reduced basis enrichment, and continue simulating Algorithm \ref{algo:full}.
This does not affect the ergodicity of Algorithm \ref{algo:full}, as will be discussed in Section \ref{sec:convergence}.
However, the error of the resulting reduced-order model can be large in the subregion of the support of the posterior that remains unexplored. 
Consequently, this large reduced-order model error may decrease the second-stage acceptance rate and the statistical efficiency of the algorithm. 

A wide range of proposal distributions can be used in Algorithm \ref{algo:full}.
In this work, we use the grouped-components adaptive Metropolis \cite{Cui_2010}, which is a variant of adaptive Metropolis \cite{Haario_2001}.
More advanced algorithms such as stochastic Newton \cite{Martin_2012} or manifold MCMC \cite{Girolami_2011} can also be used within Algorithm \ref{algo:full}.
The computational cost of evaluating derivative information for these algorithms can also be reduced by using the reduced-order models.

\subsubsection{Ergodicity of the full target algorithm.}
\label{sec:convergence}

Throughout this work, the following assumption on the forward model and the reduced-order model is used:
\begin{assumption}
\label{assum:1}
The forward model $F(\vecx)$, and the reduced-order model $F_{m}(\vecx)$ are Lipschitz continuous and bounded on $\mathbb{X}$.
\end{assumption}
Since the data-misfit functions $\Phi(\vecx)$ and $\Phi_{m}(\vecx)$ are quadratic, we have $\Phi(\vecx) \geq 0$ and $\Phi_{m}(\vecx) \geq 0$.
Assumption \ref{assum:1} implies that there exists a constant $K > 0$ such that
\(
\Phi(\vecx) \leq K  \; {\rm and} \; \Phi_{m}(\vecx) \leq K, \ \forall \vecx \in \mathbb{X}.
\)
We also note that Lipschitz continuity of $F(\vecx)$ implies the Lipschitz continuity of $\Phi(\vecx)$. Similarly, $\Phi_{m}(\vecx)$ is Lipschitz continuous.

We first establish the ergodicity of a non-adaptive version of Algorithm \ref{algo:full}, where the reduced-order model is assumed to be given and fixed.

\begin{lemma}
\label{thm:3}
In the first stage of Algorithm \ref{algo:full}, suppose the proposal distribution $q(\vecx, \vecy)$ is $\pi$-irreducible. Then the non-adaptive version of Algorithm \ref{algo:full} is ergodic.
\end{lemma}
\begin{proof}
The detailed balance condition and aperiodicity condition are satisfied by Algorithm \ref{algo:full}, see \cite{CF_2005, Liu_1998}.
Since $\exp(-\Phi_{m}(\vecy)) > 0$ for all $\vecy \in \mathbb{X}$ by Assumption \ref{assum:1}, we have $q(\vecx, \vecy) > 0$ for all $\vecx, \vecy \in \mathbb{X}$, thus the irreducibility condition is satisfied, see \cite{CF_2005}.
Ergodicity follows from the detailed balance, aperiodicity, and irreducibility.
\end{proof}

The adaptation used in Algorithm \ref{algo:full} is different from that of standard adaptive MCMC algorithms such as \cite{Haario_2001} and the original adaptive delayed acceptance algorithm \cite{Cui_2010}.
These algorithms carry out an infinite number of adaptations because modifications to the proposal distribution or to the approximation continue throughout MCMC sampling.
The number of adaptations in Algorithm \ref{algo:full}, on the other hand, cannot exceed the maximum allowable dimension $M$ of the reduced basis, and hence it is finite.
Furthermore, the adaptation stops after finite time because of the finite adaptation criterion. Since Algorithm \ref{algo:full} is ergodic for each of the reduced-order models constructed by Lemma \ref{thm:3}, Proposition 2 in \cite{Roberts_2007} ensures the ergodicity of Algorithm \ref{algo:full}. 
Lemma \ref{thm:3} also reveals that Algorithm \ref{algo:full} always converges to the full posterior distribution, regardless of the accuracy of the reduced-order model.

\section{Properties of the approximation}

In Algorithm \ref{algo:full}, the accuracy of the reduced-order model is adaptively improved during MCMC sampling, and thus it is important to analyze the error of the approximate posterior induced by the reduced-order model, compared to the full posterior.
Since bounds on the bias of the resulting Monte Carlo estimator can be derived from the error of the approximate posterior, this analysis is particularly useful in situations where we want to utilize the approximate posterior directly for further  CPU time reduction. 
This error analysis also justifies the finite adaptation criterion (see Definition \ref{fin_a}) used in Algorithm \ref{algo:full} to terminate the enrichment of the reduced basis. 

We provide here an analysis of the Hellinger distance between the full posterior distribution and the approximate posterior distribution
\begin{equation}
\dhell (\pi, \pi_{m}) = \left( \frac{1}{2} \int_\mathbb{X} \left( \sqrt{\pi(\vecx|\data)} - \sqrt{\pi_{m}(\vecx|\data)}\right)^2 d\vecx \right)^{\frac{1}{2}}.
\label{eq:dhell}
\end{equation}
The Hellinger distance translates directly into bounds on expectations, and hence we use it as a metric to quantify the error of the approximate posterior distribution.

The framework set by \cite{Stuart_2010} is adapted here to analyze the Hellinger distance between the full posterior distribution and its approximation induced by the $F_m$ constructed in Algorithm \ref{algo:full}, with respect to a given threshold $\epsilon$.
Given a set of samples $\{\vecx^{(0)}, \ldots, \vecx^{(m)}\}$ where the snapshots are computed, and the associated reduced basis $\matV_{m}$, we define the $\epsilon$-feasible set and the associated posterior measure:
\begin{definition}
\label{def:feasible}
For a given $\epsilon > 0$ and a reduced-order model $F_{m}(\vecx)$, define the $\epsilon$-feasible set as
\begin{equation}
\displaystyle \Omega^{(m)}(\epsilon) = \left\{ \vecx \in \mathbb{X} \mid \| \Sigma_{e}^{-\frac12} ( F(\vecx) - F_{m}(\vecx) ) \|_\infty \leq \epsilon \right\}.
\end{equation}
The set $\Omega^{(m)}(\epsilon) \subseteq \mathbb{X}$ has posterior measure
\begin{equation}
\mu\left(\Omega^{(m)}(\epsilon)\right) = \int_{\Omega^{(m)}(\epsilon)} \pi(\vecx | \data) d\vecx.
\end{equation}
The complement of the $\epsilon$-feasible set is given by $\Omega^{(m)}_\perp(\epsilon) = \mathbb{X} \setminus \Omega^{(m)}(\epsilon)$, which has posterior measure $\mu(\Omega^{(m)}_\perp(\epsilon)) = 1 - \mu\left(\Omega^{(m)}(\epsilon)\right)$.
\end{definition}

On any measurable subset of the $\epsilon$-feasible set, the error of the reduced-order model is bounded above by $\epsilon$.
Then the Hellinger distance (\ref{eq:dhell}) can be bounded by the user-specified $\epsilon$ and the posterior measure of $\Omega^{(m)}_\perp(\epsilon)$, which is the region that has not been well approximated by the reduced-order model.
We formalize this notion though the following propositions and theorems.

\begin{proposition}
\label{prop:1}
Given a reduced-order model $F_{m}(\vecx)$ and some $\epsilon > 0$, there exists a constant $K > 0$ such that
\(
\left| \Phi(\vecx) - \Phi_{m}(\vecx) \right| \leq K \epsilon, \ \forall \vecx \in \Omega^{(m)}(\epsilon).
\)
\end{proposition}
\begin{proof}
This result directly follows from the definition of the $\epsilon$-feasible set and Assumption \ref{assum:1}.
\end{proof}
\begin{proposition}
\label{prop:2}
Given a reduced-order model $F_{m}(\vecx)$ and a $\epsilon > 0$, there exist constants $K_1 > 0$ and $K_2 > 0$ such that
\(
| Z - Z_{m} | \leq K_1 \epsilon + K_2 \mu(\Omega^{(m)}_\perp(\epsilon)).
\)
\end{proposition}
\begin{proof}
From the definition of the $\epsilon$-feasible set, we have a bound on the difference of the normalizing constants:
\begin{eqnarray*}
\left| Z - Z_{m} \right| & \leq & \int_{\Omega^{(m)}(\epsilon)} \left| \exp(- \Phi(\vecx)) - \exp(- \Phi_{m}(\vecx)) \right| \pi_0(\vecx) d\vecx  + \nonumber \\
&& \int_{\Omega^{(m)}_\perp(\epsilon)} \left| 1 - \exp\left(\Phi(\vecx) - \Phi_{m}(\vecx)\right) \right| Z \pi(\vecx|\data) d\vecx \ .
\label{eq:norm_ie}
\end{eqnarray*}
Since $\Phi$ and $\Phi_{m}$ are Lipschitz continuous and greater than zero, we have
\[
 \left| \exp(- \Phi(\vecx)) - \exp(- \Phi_{m}(\vecx)) \right| \leq K_3 \left| \Phi(\vecx) - \Phi_{m}(\vecx) \right|,
\]
for some constant $K_3 > 0$.
Then by Proposition \ref{prop:1}, there exists a constant $K_1 > 0$ such that
\[ \int_{\Omega^{(m)}(\epsilon)} \left| \exp(- \Phi(\vecx)) - \exp(- \Phi_{m}(\vecx)) \right| \pi_0(\vecx) d\vecx \leq K_1 \epsilon.
\]
There also exists a constant $K_2 > 0$ such that
\begin{eqnarray*}
\int_{\Omega^{(m)}_\perp(\epsilon)} \left| 1 - \exp\left( \Phi(\vecx) - \Phi_{m}(\vecx)\right) \right| Z  \pi(\vecx|\data) d\vecx
& \leq & K_2 \int_{\Omega^{(m)}_\perp(\epsilon)} \pi(\vecx|\data) d\vecx \nonumber \\
& = & K_2\mu\left(\Omega^{(m)}_\perp(\epsilon)\right).
\end{eqnarray*}
Thus, we have $| Z - Z_{m} | \leq K_1 \epsilon + K_2 \mu(\Omega^{(m)}_\perp(\epsilon))$.
\end{proof}
\begin{theorem}
\label{thm:2}
Suppose we have the full posterior distribution $\pi(\vecx|\data)$ and its approximation $\pi_{m}(\vecx|\data)$ induced by a reduced-order model $F_{m}(\vecx)$.
For a given $\epsilon > 0$, there exist constants $K_1 > 0$ and $K_2 > 0$ such that
\begin{equation}
\dhell(\pi, \pi_{m}) \leq K_1 \epsilon + K_2 \mu\left(\Omega^{(m)}_\perp(\epsilon)\right).
\end{equation}
\end{theorem}
\begin{proof}
Following Theorem 4.6 of \cite{Stuart_2010} we have
\begin{eqnarray*}
\dhell (\pi, \pi_{m}) ^2 & = & \frac{1}{2} \int_\mathbb{X} \left( \sqrt{\frac{1}{Z}\exp(- \Phi(\vecx)) } - \sqrt{\frac{1}{Z_{m}}\exp(- \Phi_{m}(\vecx))} \right)^2  \pi_0(\vecx)d\vecx \\
& \leq & I_1 + I_2,
\end{eqnarray*}
where
\begin{eqnarray*}
I_1 & = & \frac{1}{Z} \int_\mathbb{X} \left( \exp(- \frac{1}{2}\Phi(\vecx)) - \exp(- \frac{1}{2} \Phi_{m}(\vecx)) \right)^2  \pi_0(\vecx) d\vecx , \\
I_2 & = & | Z^{-\frac{1}{2}} - {Z_{m}}^{-\frac{1}{2}}| ^2 Z_{m} .
\end{eqnarray*}
Following the same derivation as Proposition \ref{prop:2} we have
\begin{eqnarray*}
I_1 & \leq & \frac{1}{Z} \int_{\Omega^{(m)}} \left( \exp(- \frac{1}{2}\Phi(\vecx)) - \exp(- \frac{1}{2}\Phi_{m}(\vecx)) \right)^2 \pi_0(\vecx) d\vecx  + \\
&& \int_{\Omega^{(m)}_\perp} \left( 1 - \exp\left( \frac{1}{2}\Phi(\vecx) - \frac{1}{2}\Phi_{m}(\vecx)\right) \right)^2 \pi(\vecx|\data) d\vecx.
\end{eqnarray*}
Thus there exist constants $K_3, K_4 > 0$ such that
\(
I_1 \leq K_3 \epsilon^2 + K_4 \mu(\Omega^{(m)}_\perp(\epsilon)).
\)

Applying the bound
\[
| Z^{-\frac{1}{2}} - {Z_{m}}^{-\frac{1}{2}}|^2 \leq K \max( Z^{-3}, (Z_{m})^{-3}) | Z - Z_{m}|^2,
\]
and Proposition \ref{prop:2}, we have
\[
I_2 \leq \left( K_5 \epsilon +  K_6 \mu\left(\Omega^{(m)}_\perp(\epsilon)\right) \right)^2
\]
for some constants $K_5, K_6 > 0$.

Combining the above results we have
\[
\dhell (\pi, \pi_{m}) ^2 \leq K_3 \epsilon^2 + K_4 \mu\left(\Omega^{(m)}_\perp(\epsilon)\right) + \left( K_5 \epsilon +  K_6 \mu\left(\Omega^{(m)}_\perp(\epsilon)\right) \right)^2.
\]
Since $\epsilon > 0$ and $\mu(\Omega^{(m)}_\perp(\epsilon)) \geq 0$, the above inequality can be rearranged as
\[
\dhell (\pi, \pi_{m}) ^2 \leq \left( K_1 \epsilon +  K_2 \mu\left(\Omega^{(m)}_\perp(\epsilon)\right) \right)^2,
\]
for some constants $K_1 > 0$ and $K_2 > 0$.
\end{proof}

Under certain technical conditions \cite{Bui_2008_2, Patera_2007}, the pointwise error of the reduced-order model decreases asymptotically as the dimensionality of the reduced basis increases. 
Thus the $\epsilon$-feasible set asymptotically grows with the reduced basis enrichment, and hence $\mu(\Omega^{(m)}_\perp(\epsilon))$ asymptotically decays. 
If the posterior distribution is sufficiently well-sampled such that
\begin{equation}
\mu\left(\Omega^{(m)}_\perp(\epsilon)\right) \leq \epsilon,
\label{eq:well_sampled}
\end{equation}
then the Hellinger distance (\ref{eq:dhell}) is characterized entirely by $\epsilon$, as shown in Theorem \ref{thm:2}.
Thus, by adaptively updating the data-driven reduced-order model until the condition (\ref{eq:well_sampled}) is satisfied, we can build an approximate posterior distribution whose error is proportional to the user-specified error threshold $\epsilon$.

In practice, it is not feasible to check the condition (\ref{eq:well_sampled}) within MCMC sampling, so we use heuristics to motivate the finite adaptation criterion in Definition \ref{fin_a}. Consider a stationary Markov chain that has invariant distribution $\pi(\vecx|\data)$. We assume that the probability of visiting $\Omega^{(m)}_\perp(\epsilon)$ is proportional to its posterior measure. 
Suppose we have $\mu(\Omega^{(m)}_\perp(\epsilon)) = \epsilon$, then the probability of the Markov chain visiting $\Omega^{(m)}_\perp(\epsilon)$ is about $c \epsilon$, for some $c > 0$.
In this case, the average number of MCMC steps needed for the next reduced basis refinement is about $\frac{1}{c \epsilon}$.
Since the posterior measure of $\Omega^{(m)}_\perp(\epsilon)$ decays asymptotically with refinement of the reduced basis, the average number of steps needed for the basis refinement asymptotically increases.
Thus we can treat the condition (\ref{eq:well_sampled}) as holding if the average number of steps used for reduced basis refinement exceeds $N_{\rm max} = \frac{1}{c \epsilon}$, and thus terminate the adaptation.
We recommend to choose a small $c$ value, for example $c = 0.1$, to delay the stopping time of the adaptation.
In this way, the adaptive construction process can search the parameter space more thoroughly to increase the likelihood that the condition $\mu(\Omega^{(m)}_\perp(\epsilon)) < \epsilon$ is satisfied.

\section{Extension to approximate Bayesian inference}
\label{sec:approx_MCMC}
The full target algorithm introduced in Section \ref{sec:full_algo} has to evaluate the forward model after each subchain simulation to preserve ergodicity.
The resulting Monte Carlo estimator (\ref{eq:mc}) is unbiased. However, the effective sample size (ESS) for a given budget of computing time is still characterized by the number of full model evaluations.
To increase the ESS, we also consider an alternative approach that directly samples from the approximate posterior distribution when the reduced-order model has sufficient accuracy.

Suppose that the scaled error indicator $\widehat{t}(\vecx)$ provides a reliable estimate of the scaled true error of the reduced-order model.
Then the reliability and the refinement of the reduced-order model can be controlled by the error indicator.
During MCMC sampling, we only evaluate the full model to correct a Metropolis acceptance and to update the reduced-order model when the error indicator exceeds the threshold $\epsilon$.
When a reduced-order model evaluation has error indicator less than $\epsilon$, we treat the reduced-order model as a sufficient approximation to the full model. In this case, decisions in the MCMC algorithm are based on the approximate posterior distribution.

Compared to the full target algorithm that draws samples from the full posterior distribution, this approach only samples from an approximation that mimics the full posterior up to the error threshold $\epsilon$.
We refer to the proposed approach as the ``$\epsilon$-approximate algorithm.''
Even though sampling the approximate posterior distribution will introduce bias to the Monte Carlo estimator (\ref{eq:mc}), this bias may be acceptable if the resulting Monte Carlo estimator has a smaller mean squared error (MSE) for a given amount of computational effort, compared to the standard Metropolis-Hasting algorithm.
In the remainder of this section, we provide details on the algorithm and analyze the bias of the resulting Monte Carlo estimator.

\subsection{$\epsilon$-approximate algorithm}
Algorithm \ref{algo:approx} details the $\epsilon$-approximate algorithm.
During the adaptation, for a proposed candidate $\vecx^\prime$, we discard the reduced-order model evaluation if its error indicator exceeds the new upper threshold $\epsilon_0$.
We set $\epsilon_0 = 1$, which means that Algorithm \ref{algo:approx} does not use the information from the reduced-order model if its estimated error is greater than one standard deviation of the measurement noise.
In this case (Lines 3--6), we run the full model directly to evaluate the posterior distribution and the acceptance probability.
If the error indicator of the reduced-order model is between the lower threshold $\epsilon$ and the upper threshold $\epsilon_0$ (Lines 7--15), the reduced-order model is considered to be a reasonable approximation, and the delayed acceptance scheme is used to make the correction.
If the error indicator is less than the lower threshold $\epsilon$, or if the adaptation is stopped (Lines 17--19), the reduced-order model is considered to be a sufficiently accurate approximation to the full model, and is used to accept/reject the proposal directly.

The adaptation criterion used in Lines 2 of Algorithm \ref{algo:approx} has two conditions: the dimension of reduced basis should not exceed a specified threshold $M$, and the finite adaptation criterion (Definition~\ref{fin_a}) should not yet be triggered.
The reduced basis is updated using the full model evaluations at proposals accepted by MCMC.

When the reduced basis dimension reaches $M$ but the finite adaptation criterion is not satisfied, it is not appropriate to use the $\epsilon$-approximate algorithm for the prescribed error threshold. This is because the large reduced-order model error can potentially result in unbounded bias in the Monte Carlo estimator. 
In this situation, we should instead use the full target algorithm, for which convergence is guaranteed.

\begin{algorithm}[h!]
\caption{$\epsilon$-approximate algorithm}
\label{algo:approx}
\begin{algorithmic}[1]
  \REQUIRE Given the subchain length $L$, the maximum allowable reduced basis dimension $M$, the upper threshold $\epsilon_0$, and the error threshold $\epsilon$. At step $n$, given state $X_n = \vecx$, a proposal $q(\vecx, \cdot)$, and a reduced-order model $F_{m}(\cdot)$ defined by the reduced basis $\matV_{m}$, one step of the algorithm is:
  \STATE Propose $\vecx^{\prime} \sim q(\vecx, \cdot)$, then evaluate the reduced-order model $F_{m}(\vecx^{\prime})$ and  $\widehat{\vect}_m(\vecx^{\prime})$
  \IF{{\bf finite adaptation criterion} is not satisfied \AND $m < M$ \AND $\norm{\widehat{\vect}_m(\vecx^{\prime})}_{\infty} \geq \epsilon$}
  \IF{$\norm{\widehat{\vect}_m(\vecx^{\prime})}_{\infty} \geq \epsilon_0$}%
    \STATE Discard the reduced-order model evaluation, and compute the acceptance probability using the full posterior distribution
  \[
  \alpha(\vecx, \vecx^{\prime}) = 1 \wedge \frac{\pi(\vecx^{\prime}|\data)}{\pi(\vecx|\data)} \frac{q(\vecx^{\prime}, \vecx)}{q(\vecx, \vecx^{\prime})}
  \]
    \STATE Accept/reject $\vecx^{\prime}$ according to ${\rm Uniform}(0, 1] < \alpha(\vecx, \vecx^{\prime})$
  	\STATE Update the reduced basis $\matV_{m}$ to $\matV_{m+1}$ using the new full model evaluation at accepted $\vecx^{\prime}$ by a Gram-Schmidt process
  \ELSE[$\epsilon \leq \norm{\widehat{\vect}_m(\vecx^{\prime})}_{\infty} < \epsilon_0$]
    \STATE Run the delayed acceptance for 1 step, evaluate the acceptance probability
  \[
  \beta_1(\vecx, \vecx^{\prime}) = 1 \wedge \frac{\pi_{m}(\vecx^{\prime}|\data)}{\pi_{m}(\vecx|\data)} \frac{q(\vecx^{\prime}, \vecx)}{q(\vecx, \vecx^{\prime})}
  \]
  	 \IF{${\rm Uniform}(0, 1] < \beta_1(\vecx, \vecx^{\prime})$}
 	   \STATE Run the full model at $\vecx^{\prime}$ to evaluate the full posterior and the acceptance probability
  \[
  \beta_2(\vecx, \vecx^{\prime}) = 1 \wedge \frac{\pi(\vecx^{\prime}|\data)}{\pi(\vecx|\data)} \frac{\pi_{m}(\vecx|\data)}{\pi_{m}(\vecx^{\prime}|\data)}
  \]
      \STATE Accept/reject $\vecx^{\prime}$ according to ${\rm Uniform}(0, 1] < \beta_2(\vecx, \vecx^{\prime})$
  	 \ELSE
  	   \STATE Reject $\vecx^{\prime}$ by setting $X_{n+1} = \vecx$
  	 \ENDIF
      \STATE Update the reduced basis $\matV_{m}$ to $\matV_{m+1}$ as in Line 6
  \ENDIF
  \ELSE[$\norm{\widehat{\vect}_m(\vecx^{\prime})}_{\infty} < \epsilon$ \OR Adaptation is stopped]
    \STATE The reduced-order model is used directly to evaluate the acceptance probability
  \[
  \theta(\vecx, \vecx^{\prime}) = 1 \wedge \frac{\pi_{m}(\vecx^{\prime}|\data)}{\pi_{m}(\vecx|\data)} \frac{q(\vecx^{\prime}, \vecx)}{q(\vecx, \vecx^{\prime})}
  \]
    \STATE Accept/reject $\vecx^{\prime}$ according to ${\rm Uniform}(0, 1] < \theta(\vecx, \vecx^{\prime})$
  \ENDIF
\end{algorithmic}
\end{algorithm}

\subsection{Monte Carlo error of the $\epsilon$-approximate algorithm}
To analyze the performance of the $\epsilon$-approximate algorithm, we compare the MSE of the resulting Monte Carlo estimator with that of a standard single-stage MCMC algorithm that samples the full posterior distribution.

We wish to compute the expectation of a function $h(\vecx)$ over the posterior distribution $\pi(\vecx | \data)$, i.e.,
\begin{equation}
I(h) = \int_\mathbb{X} h(\vecx) \pi(\vecx|\data) d\vecx,
\end{equation}
where the first and second moments of $h(\vecx)$ are assumed to be finite.
Suppose a single-stage MCMC algorithm can sample the full posterior distribution for $N_1$ steps in a fixed amount of CPU time, and $\ess(h)$ effective samples are produced.
The resulting Monte Carlo estimator
\begin{equation}
\widehat{I(h)} = \frac{1}{N_1} \sum_{i = 1}^{N_1} h(\vecx^{(i)}), \quad \vecx^{(i)} \sim \pi(\cdot|\data),
\end{equation}
has MSE
\begin{equation}
\mse\left(\widehat{I(h)}\right) = \frac{\var(h)}{\ess(h)},
\end{equation}
which is characterized by the ESS and the variance of $h$ over $\pi(\vecx | \data)$.

By sampling the approximate posterior distribution, the expectation $I(h)$ can be approximated by
\begin{equation}
I_{m}(h) = \int_\mathbb{X} h(\vecx) \pi_{m}(\vecx|\data) d\vecx.
\end{equation}
Suppose we can sample the approximate posterior for $N_2$ steps in the same amount of CPU time as sampling the full posterior, and that these $N_2$ samples have effective sample size
\begin{equation}
\ess_{m}(h) = S(m) \times \ess(h),
\end{equation}
where $S(m) > 1$ is the speedup factor that depends on the computational expense of the reduced-order model $F_m(\vecx)$.
The Monte Carlo estimator
\begin{equation}
\widehat{I_{m}(h)} = \frac{1}{N_2} \sum_{i = 1}^{N_2} h(\vecx^{(i)}), \quad \vecx^{(i)} \sim \pi_{m}(\cdot|\data),
\end{equation}
has the MSE
\begin{equation}
\mse\left(\widehat{I_{m}(h)}\right) = \frac{\var_{m}(h)}{\ess_{m}(h)} + \bias\left( \widehat{I_{m}(h)} \right)^2,
\end{equation}
where the bias is defined by %
\begin{equation}
\bias\left( \widehat{I_{m}(h)} \right) = I_{m}(h) - I(h).
\end{equation}

We are interested in the situation for which sampling the approximation leads to a smaller MSE than sampling the full posterior distribution, i.e.,
\begin{equation}
\frac{\var_{m}(h)}{\ess_{m}(h)} + \bias\left( \widehat{I_{m}(h)} \right)^2 \leq \frac{\var(h)}{\ess(h)} .
\end{equation}
Rearranging the above inequality gives
\begin{equation}
\ess(h) \leq \tau := \frac{\var(h)}{\bias( \widehat{I_{m}(h)} )^2} \left( 1 - \frac{\var_{m}(h)}{\var(h)} \frac{1}{S(m)} \right) .
\label{eq:ess_bound}
\end{equation}
Equation (\ref{eq:ess_bound}) reveals that when our target ESS for drawing from the full posterior distribution does not exceed the bound $\tau$, sampling the approximate posterior will produce a smaller MSE.
This suggests that the $\epsilon$-approximate algorithm will be more accurate for a fixed computational cost than the single-stage MCMC when the target ESS of the single-stage MCMC satisfies (\ref{eq:ess_bound}).
In such cases, the MSE of the Monte Carlo estimator is dominated by the variance rather than the bias.

The bias is characterized by the Hellinger distance between the full posterior distribution and its approximation (Lemma 6.37 of \cite{Stuart_2010}), i.e.,
\[
\bias\left( \widehat{I_{m}(h)} \right)^2 \leq 4 \left ( \int h(\vecx)^2 \pi(\vecx|\data)d\vecx +  \int h(\vecx)^2 \pi_{m}(\vecx|\data)d\vecx \right) \dhell (\pi, \pi_{m})^2 .
\]
We assume that there exists an $m^{\ast} = m(\epsilon)$ and a set of samples $\{ \vecx^{(0)}, \ldots, \vecx^{(m^{\ast})} \}$ such that the resulting reduced-order model $F_{m^{\ast}}(\vecx)$ satisfies the condition (\ref{eq:well_sampled}), i.e., $\mu(\Omega^{(m^{\ast})}_\perp(\epsilon)) \leq \epsilon$.
Applying Theorem \ref{thm:2}, the ratio of variance to squared bias can be simplified to
\[
\frac{\var(h)}{\bias( \widehat{I_{m(\epsilon)}(h)} )^2} \geq \frac{ K}{ \epsilon^2 },
\]
for some constant $K>0$, and hence we have
\begin{equation}
\tau \geq \frac{ K}{\epsilon^2 } \left( 1 - \frac{\var_{m(\epsilon)}(h)}{\var(h) S(m(\epsilon))} \right) .
\label{eq:ess_bound_3}
\end{equation}
For problems that have reliable error indicators or estimators, the $\epsilon$-approximate algorithm provides a viable way to select the set of $m^{\ast} = m(\epsilon)$ samples for computing the snapshots.
However it is computationally infeasible to verify that the condition (\ref{eq:well_sampled}) holds in practice.
We thus employ the finite adaptation criterion (Definition \ref{fin_a}) to perform a heuristic check on the condition (\ref{eq:well_sampled}), as discussed in Section 4.

The bound $\tau$ is characterized by the user-given error threshold $\epsilon$ and the speedup factor $S(m(\epsilon))$, where $S(m(\epsilon))$ is a problem-specific factor that is governed by the convergence rate and computational complexity of the reduced-order model.
For a reduced-order model such that $S(m(\epsilon)) > \var_{m(\epsilon)}(h)/\var(h)$, there exists a $\tau$ so that the MSE of sampling the approximate posterior for $\tau \times S(m(\epsilon))$ steps will be less than the MSE of sampling the full posterior for the same amount of CPU time.
In the regime where the reduced-order models have sufficiently large speedup factors, the bound $\tau$ is dominated by $\epsilon^2$, and hence decreasing $\epsilon$ results in a higher bound $\tau$.
However there is a trade-off between the numerical accuracy and speedup factors.
We should avoid choosing a very small $\epsilon$ value, because this can potentially lead to a high-dimensional reduced basis and a correspondingly expensive reduced-order model such that $S(m(\epsilon)) < \var_{m(\epsilon)}(h)/\var(h)$, where the  ratio $\var_{m(\epsilon)}(h)/\var(h)$ should be close to one for such an accurate reduced-order model. 
In this case, sampling the approximate posterior can be less efficient than sampling the full posterior.

\section{Numerical results and discussion}
To benchmark the proposed algorithms, we use a model of isothermal steady flow in porous media, which is a classical test case for inverse problems.
\subsection{Problem setup}
Let $D = [0, 1]^2$ be the problem domain, $\partial D$ be the boundary of the domain, and $r \in D$ denote the spatial coordinate.
Let $k(r)$ be the unknown permeability field, $u(r)$ be the pressure head, and $q(r)$ represent the source/sink.
The pressure head for a given realization of the permeability field is governed by
\begin{equation}
\nabla \cdot \left( k(r) \nabla u(r) \right) +  q(r) = 0, \quad r \in D,
\label{eq:forward}
\end{equation}
where the source/sink term $q(r)$ is defined by the superposition of four weighted Gaussian plumes with standard deviation $0.05$, centered at $r=[0.3, 0.3]$, $[0.7, 0.3]$, $[0.7, 0.7]$, $[0.3, 0.7]$, and with weights $\{2, -3, -2, 3\}$.
A zero-flux Neumann boundary condition 
\begin{equation}
k(r) \nabla u(r) \cdot \vec{n}(r)  = 0, \quad  r \in \partial D,
\label{eq:bnd1} 
\end{equation}
is prescribed, where $\vec{n}(r)$ is the outward normal vector on the boundary.
To make the forward problem well-posed, we impose the extra boundary condition
\begin{equation}
\int_{\partial D} u(r) d l(r)= 0.
\label{eq:bnd2}
\end{equation}
Equation (\ref{eq:forward}) with boundary conditions (\ref{eq:bnd1}) and (\ref{eq:bnd2}) is solved by the finite element method with $120\times 120$ linear elements.
This leads to the system of equations (\ref{eq:forward_a}).

In Section \ref{sec:9d}, we use a nine-dimensional example to carry out numerical experiments to benchmark various aspects of our algorithms.
In this example, the spatially distributed permeability field is projected onto a set of radial basis functions, and hence inference is carried out on the weights associated with each of the radial basis functions.
In Section \ref{sec:hd_log}, we apply our algorithms to a higher dimensional problem, where the parameters are defined on the computational grid and endowed with a Gaussian process prior.
Both examples use fixed Gaussian proposal distributions, where the covariances of the proposals are estimated from a short run of the $\epsilon$-approximate algorithm.
In Section~\ref{sec:end_nr}, we offer additional remarks on the performance of the $\epsilon$-approximate algorithm.

\subsection{The 9D inverse problem}
\label{sec:9d}
\begin{figure}[h]
\centering
\includegraphics[width=\textwidth]{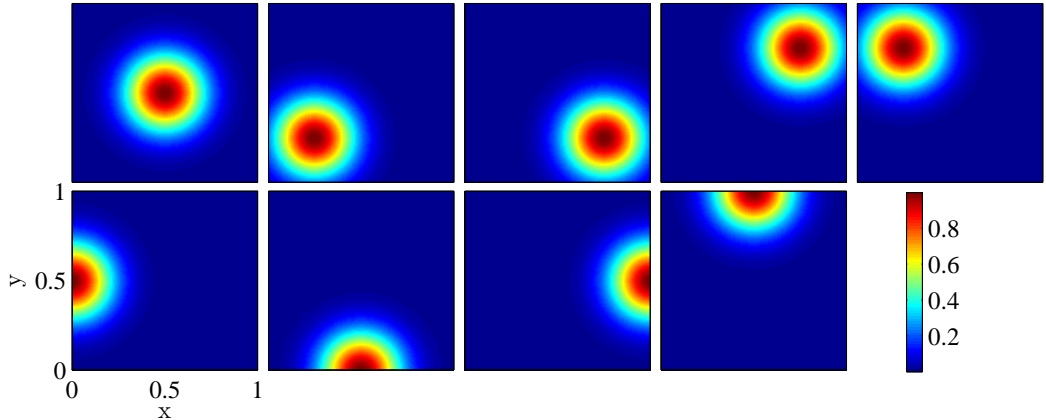}
\caption{Radial basis functions used to define the permeability field in the nine-dimensional example.}
\label{fig:param_9D}
\end{figure}
\begin{figure}[h]
\centering
\includegraphics[width=\textwidth]{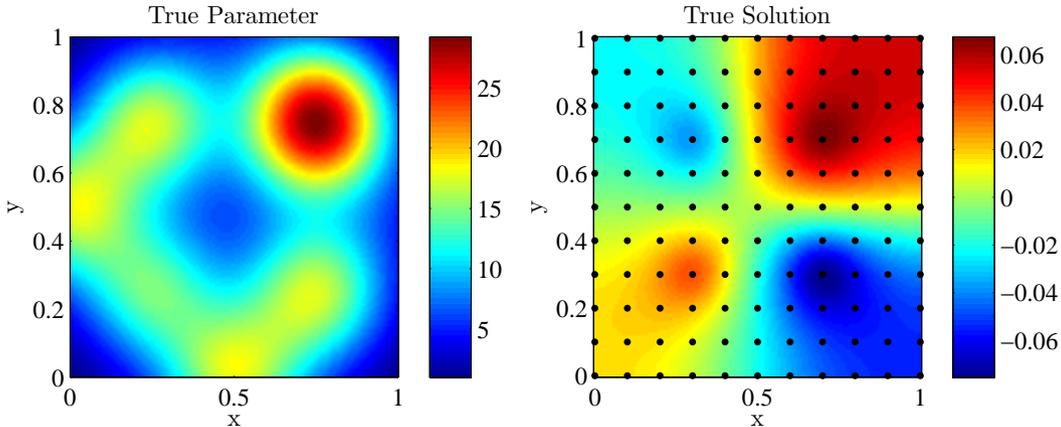}
\caption{Setup of the test case for the nine-dimensional example. Left: the true permeability used for generating the synthetic data sets. Right: the model outputs of the true permeability. The black dots indicate the measurement sensors.}
\label{fig:setup_9D}
\end{figure}
The permeability field is defined by $N_{\rm p} = 9$ radial basis functions:
\begin{equation}
k(r) = \sum_{i = 1}^{N_{\rm p}} b(r; r_i) x_i, \;\;
b(r; r_i) = \exp \left[- 0.5 \left( \frac{\norm{r-r_i}}{0.15} \right)^{2} \right],
\end{equation}
where $r_1, \ldots, r_9$ are the centers of the radial basis functions. These radial basis functions are shown in Figure~\ref{fig:param_9D}.
The prior distributions on each of the weights $x_i$, $i=1,\ldots,9$ are independent and log-normal, and hence we have
\begin{equation}
\pi_0(\vecx) \propto \prod_{i=1}^{N_{\rm p}} \exp\left (- \frac{{\log(x_i)}^2}{2{\sigma_0}^2} \right),
\label{eq:log-normal}
\end{equation}
where $\sigma_0 = 2$ and $N_{\rm p}$ = 9.
The true permeability field used to generate the test data, and the corresponding pressure head are shown in Figure~\ref{fig:setup_9D}. The measurement sensors are evenly distributed over $D$ with grid spacing $0.1$, and the signal-to-noise ratio of the observed data is $50$.

Numerical experiments for various choices of  $\epsilon$ are carried out to test the computational efficiency of both algorithms and the dimensionality of reduced basis in the reduced-order model.
For $\epsilon = \{10^{-1}, 10^{-2}, 10^{-3}\}$, we run the full target algorithm for $10^4$ iterations, with subchain length $L = 50$. 
To make a fair comparison in terms of the number of posterior evaluations, we run the $\epsilon$-approximate algorithm for $5\times 10^5$ iterations, also with $\epsilon = \{10^{-1}, 10^{-2}, 10^{-3}\}$.
For both algorithms, the reduced-order model construction process is started at the beginning of the MCMC simulation. 
We set $c = 10^{-1}$ in the finite adaptation criterion \ref{fin_a}.
As a reference, we run a single-stage MCMC algorithm for $5\times 10^5$ iterations using the same proposal distribution.
The reference algorithm only uses the full posterior distribution.
The first $2000$ samples of the simulations generated by the full target algorithm are discarded as burn-in samples. Similarly, the first $10^5$ samples (to match the number of reduced-order model evaluations in the full target algorithm) are discarded for the simulations using the $\epsilon$-approximate algorithm and the reference algorithm.

In the remainder of this subsection, we provide various benchmarks of our data-driven model reduction approach and the two algorithms. These include:
\begin{enumerate}
\item A comparison of the full target algorithm and the $\epsilon$-approximate algorithm with the reference algorithm.
\item A comparison of the data-driven reduced-order model with the reduced-order model built with respect to the prior distribution.
\item A demonstration of the impact of observed data on our data-driven reduced-order model.
\end{enumerate}

\begin{table}[h]
\caption{Comparison of the computational efficiency of the full target algorithm with $\epsilon = \{10^{-1}, 10^{-2}, 10^{-3}\}$ and the $\epsilon$-approximate algorithm with $\epsilon = \{10^{-1}, 10^{-2}, 10^{-3}\}$ with the reference algorithm. The second step acceptance probability $\beta$ is defined in Algorithm \ref{algo:full}.  The posterior measure of the complement of the $\epsilon$-feasible set, $\mu(\Omega^{(m)}_\perp(\epsilon))$, is given in Definition \ref{def:feasible}.}
\label{table:performance_a}
\centering
\tabsize
\begin{tabular}{L{3cm}L{1.5cm}L{1.5cm}L{1.5cm}L{1.5cm}L{1.5cm}L{1.5cm}L{1.5cm}} \toprule
& \lower.3ex\hbox{Reference} & \multicolumn{3}{c}{\lower.3ex\hbox{Full target}} & \multicolumn{3}{c}{\lower.3ex\hbox{$\epsilon$-approximate}} \\ \midrule
\lower.3ex\hbox{Error threshold $\epsilon$} & \lower.3ex\hbox{-} & \lower.3ex\hbox{$10^{-1}$} & \lower.3ex\hbox{$10^{-2}$} & \lower.3ex\hbox{$10^{-3}$}& \lower.3ex\hbox{$10^{-1}$} & \lower.3ex\hbox{$10^{-2}$} & \lower.3ex\hbox{$10^{-3}$} \\ 
\lower.3ex\hbox{Average $\beta$} & \lower.3ex\hbox{-}  & \lower.3ex\hbox{$0.97$} & \lower.3ex\hbox{$0.98$} & \lower.3ex\hbox{$0.98$} & \lower.3ex\hbox{-} & \lower.3ex\hbox{-} & \lower.3ex\hbox{-} \\ 
\lower.3ex\hbox{Full model evaluations} & \lower.3ex\hbox{$5 \times 10^5$} & \lower.3ex\hbox{$10^4$} & \lower.3ex\hbox{$10^4$} & \lower.3ex\hbox{$10^4$} & \lower.3ex\hbox{$13$} & \lower.3ex\hbox{$33$} & \lower.3ex\hbox{$57$} \\ 
\lower.3ex\hbox{Reduced basis vectors} & - & \lower.3ex\hbox{$14$} & \lower.3ex\hbox{$33$} & \lower.3ex\hbox{$57$} & \lower.3ex\hbox{$13$} & \lower.3ex\hbox{$33$} & \lower.3ex\hbox{$57$} \\ 
\lower.3ex\hbox{CPU time (sec)} & \lower.3ex\hbox{$34470$} & \lower.3ex\hbox{$754$} & \lower.3ex\hbox{$772$} & \lower.3ex\hbox{$814$} & \lower.3ex\hbox{$115$} & \lower.3ex\hbox{$138$} & \lower.3ex\hbox{$187$} \\ 
\lower.3ex\hbox{ESS} & \lower.3ex\hbox{$4709$} & \lower.3ex\hbox{$4122$} & \lower.3ex\hbox{$4157$} & \lower.3ex\hbox{$4471$} & \lower.3ex\hbox{$4672$} & \lower.3ex\hbox{$4688$} & \lower.3ex\hbox{$4834$} \\ 
\lower.3ex\hbox{ESS / CPU time} & \lower.3ex\hbox{$1.4 \times 10^{-1}$} & \lower.3ex\hbox{$5.5$} & \lower.3ex\hbox{$5.4$} & \lower.3ex\hbox{$5.5$} & \lower.3ex\hbox{$40.6$} & \lower.3ex\hbox{$33.9$} & \lower.3ex\hbox{$25.9$} \\ 
\lower.3ex\hbox{speedup factor} & \lower.3ex\hbox{$1$} & \lower.3ex\hbox{$40$} & \lower.3ex\hbox{$39$} & \lower.3ex\hbox{$40$} & \lower.3ex\hbox{$297$} & \lower.3ex\hbox{$248$} & \lower.3ex\hbox{$189$}  \\ 
\lower.3ex\hbox{$\mu(\Omega^{(m)}_\perp(\epsilon))$} & - & \lower.3ex\hbox{$0.1\times 10^{-4}$} & \lower.3ex\hbox{$0.7 \times 10^{-4}$} & \lower.3ex\hbox{$0$} & \lower.3ex\hbox{$1.3\times 10^{-4}$} & \lower.3ex\hbox{$0.8 \times 10^{-4}$} & \lower.3ex\hbox{$0$}  \\ \bottomrule
\end{tabular}
\end{table}

\subsubsection{Computational efficiency.}
\label{sec:epsilon}
Table \ref{table:performance_a} summarizes the number of full model evaluations, the dimensionality of reduced basis, the CPU time, the ESS, and the speedup factor, comparing the full target algorithm and the $\epsilon$-approximate algorithm with the reference algorithm.
For the reduced-order models generated by the adaptive construction process, we provide 
 estimates of the posterior measure of the complement of the $\epsilon$-feasible set, $\mu(\Omega^{(m)}_\perp(\epsilon))$.
We also provide a summary of the average second-stage acceptance probability, $\beta$, for the full target algorithm.

For the full target algorithm, the average second-stage acceptance probabilities for all three $\epsilon$ values are greater than $0.96$ in this test case.
This shows that the reduced-order models produced by all three $\epsilon$ values are reasonably accurate compared to the full model, and hence simulating the approximate posterior distribution in the first stage usually yields the same Metropolis acceptance decision as simulating the full posterior distribution.
As we enhance the accuracy of the reduced-order model by decreasing the value of $\epsilon$, the dimensionality of the resulting reduced basis increases, and thus the reduced-order model takes longer to evaluate.
Because the full target algorithm evaluates the full model for every $50$ reduced-order model evaluations, its computational cost is dominated by the number of full model evaluations.
Thus the speedup factors for all three choices of $\epsilon$ are similar (approximately 40).
Since all three reduced-order models are reasonably accurate here, the efficiency gain of using a small $\epsilon$ value is not significant.
In this situation, one could consider simulating the subchain in the first stage for more iterations (by increasing the subchain length $L$) 
when the value of $\epsilon$ is small.

The $\epsilon$-approximate algorithm produces speedup factors that are $4.7$ to $7.4$ times higher than the speedup factor of the full target algorithm in this test case.
A larger $\epsilon$ value produces a larger speedup factor, because the dimension of the associated reduced basis is smaller.

\begin{figure}[h]
\centering
\includegraphics[width=\textwidth]{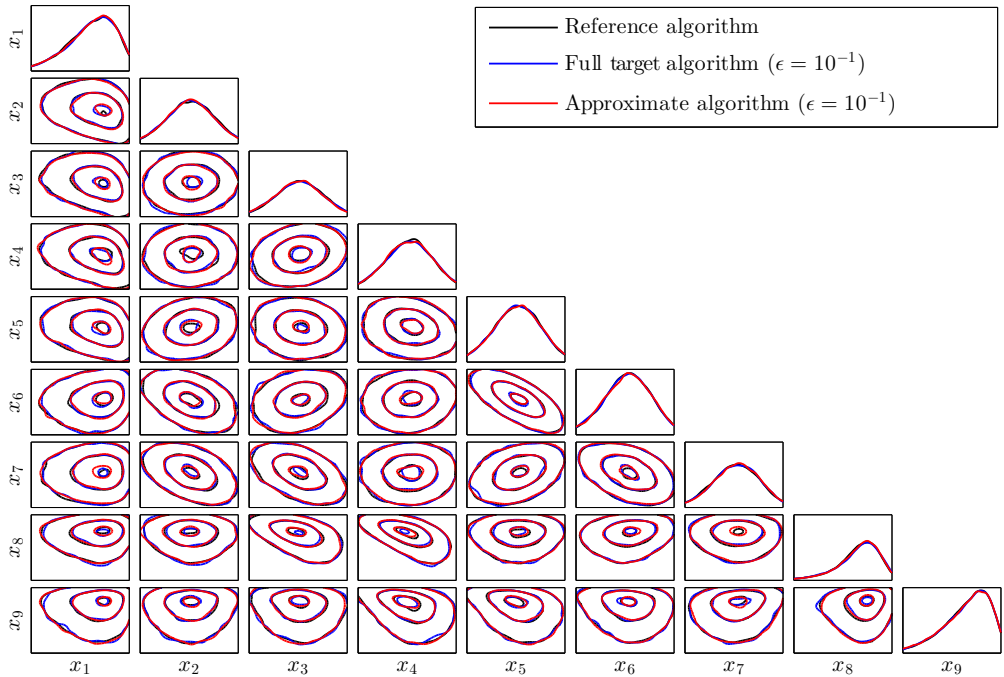}
\caption{The marginal distribution of each component of the parameter $\vecx$, and the contours of the marginal distribution of each pair of components. Black line: the reference algorithm. Blue line: the full target algorithm with $\epsilon = 10^{-1}$. Red line:  the $\epsilon$-approximate algorithm with $\epsilon = 10^{-1}$.}
\label{fig:post_2D}
\end{figure}

To assess the sampling accuracy of the $\epsilon$-approximate algorithm, Figure~\ref{fig:post_2D} provides a visual inspection of the marginal distributions of each component of the parameter $\vecx$, and the contours of the marginal distributions of each pair of components.
The black lines represent the results generated by the reference algorithm, the blue lines represent results of the full target algorithm with $\epsilon = 10^{-1}$, and red lines represent results of the $\epsilon$-approximate algorithm with $\epsilon = 10^{-1}$.
The results from the more accurate simulations that use smaller $\epsilon$ values are not shown, as they are visually close to the case $\epsilon=10^{-1}$.
The plots in Figure \ref{fig:post_2D} suggest that all the algorithms generate similar marginal distributions in this test case.
We note that both the reference algorithm and the full target algorithm sample from the full posterior distribution, and thus the small differences in the contours produced by various algorithms are likely caused by Monte Carlo error.

An alternative way to assess the computational efficiency and sampling accuracy of the $\epsilon$-approximate algorithm is to compare the number of effective samples generated by the $\epsilon$-approximate algorithm and the reference algorithm for a fixed amount of CPU time.
As shown in Table~\ref{table:performance_a}, the $\epsilon$-approximate algorithm with $\epsilon = 10^{-1}$ generates $4672$ effective samples in $115.3$ seconds; the reference algorithm can only generate about $16$ effective samples in the same amount of CPU time.
In the situation where the desired number of effective samples is at least an order of magnitude larger than the speedup factor, using the $\epsilon$-approximate algorithm is clearly advantageous to using the reference algorithm.

For both the full target algorithm and the $\epsilon$-approximate algorithm, and for each choice of $\epsilon$, we use $2\times10^{6}$ samples generated by the reference algorithm to compute the Monte Carlo estimator of the posterior measure of the complement of the $\epsilon$-feasible set for the final reduced-order model produced by the adaptive construction process.
As shown in Table~\ref{table:performance_a}, for all the reduced-order models, we have estimated $\mu(\Omega^{(m)}_\perp(\epsilon)) < \epsilon$.
This suggests that the Hellinger distances between the full posterior distribution and its approximation can be characterized by the $\epsilon$ values in all three cases, and thus the finite adaptation criterion (Definition~\ref{fin_a}) with $c = 10^{-1}$ provides a useful indicator for terminating adaptation.

For $\epsilon = 10^{-1}$, we note that the dimensions of the reduced bases produced by the full target algorithm and the $\epsilon$-approximate algorithm are different.
This is because the snapshots are evaluated at selected samples that are randomly drawn from the posterior. The spread of the sample set slightly affects the accuracy of the reduced-order model.
Nonetheless, both reduced-order models achieve the desirable level of accuracy since the estimated posterior measures $\mu(\Omega^{(m)}_\perp(\epsilon))$ are less than $\epsilon = 10^{-1}$ in this case.

Numerical experiments with $c = 10^{-2}$ and $c = 10^{-3}$ in the finite adaptation criterion \ref{fin_a} are also conducted. 
For both algorithms, choosing these smaller $c$ values leads only to one or two additional basis vectors being added in all the test cases, compared to the case $c = 10^{-1}$. 
The resulting marginal distributions generated by using $c = 10^{-2}$ and $c = 10^{-3}$ are similar to the case $c = 10^{-1}$.
For brevity, the sampling results for these experiments are not reported.
We consistently observe that the number of MCMC steps between adjacent basis enrichments increases as the adaptive construction progresses in these experiments. 
This is expected since the posterior measure $\mu(\Omega^{(m)}_\perp(\epsilon))$ asymptotically decreases with reduced basis enrichment.  
In this situation, choosing a smaller $c$ value leads only to minor increases in both of the numerical accuracy and the computational cost of the reduced-order model.
Thus the sampling accuracy and the overall computational load of both sampling algorithms are not sensitive to the smaller $c$ values in this case. 

\subsubsection{Comparison with a reduced-order model built from the prior.}
\label{sec:accuracy}
\begin{figure}[h]
\centering
\includegraphics[width=\textwidth]{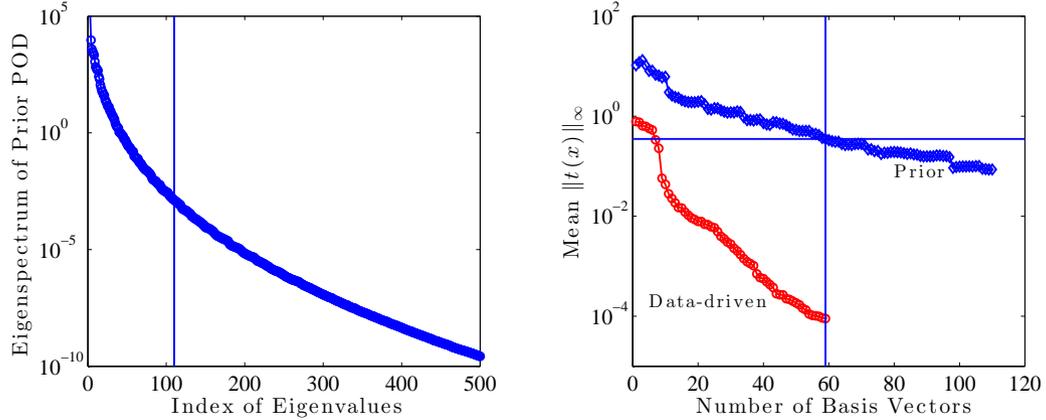}
\caption{Left: The eigenspectrum of the POD basis computed from the prior distribution. The blue line indicates the index for truncating the eigendecomposition. Right: Comparison of the numerical accuracy of the data-driven reduced-order model ($\epsilon = 10^{-3}$) with the reduced-order model built with respect to the prior distribution (\ref{eq:log-normal}). The expectation of the $L_\infty$ norm of the scaled true errors over the full posterior distribution is used as a benchmark. }
\label{fig:err}
\end{figure}

Now we compare the accuracy of the data-driven reduced-order model built with $\epsilon = 10^{-3}$ to that of a reduced-order model constructed with respect to the prior distribution
(\ref{eq:log-normal}).
To construct the reduced-order model with respect to the prior, we use proper orthogonal decomposition (POD). $10^4$ random prior samples are drawn to compute the snapshots. The POD eigenspectrum is shown in the left plot of Figure \ref{fig:err}. The eigendecomposition is truncated when it captures all but $10^{-8}$ energy (relative 2-norm of the POD eigenvalues), leading to 110 reduced basis vectors being retained in the POD basis.

By using posterior samples generated from a separate reference algorithm, we compute the expectation of the $L_\infty$ norm of the scaled true error (\ref{eq:scaled_error}) over the full posterior distribution.
The $L_\infty$ norm of the scaled true error gives the worst-case sensor error; its expectation over the the posterior quantifies the average numerical accuracy of the resulting reduced-order model.
The right plot of Figure~\ref{fig:err} shows this expectation with respect to the dimension of the reduced basis.
For this test case, the data-driven reduced-order model undergoes a significant accuracy improvement once it includes at least 10 reduced basis vectors. The figure shows that the data-driven reduced-order model has a better convergence rate compared to the reduced-order model built from the prior.

\subsubsection{The influence of posterior concentration.}
\label{sec:lkd}

\begin{figure}[h]
\centering
\includegraphics[width=\textwidth]{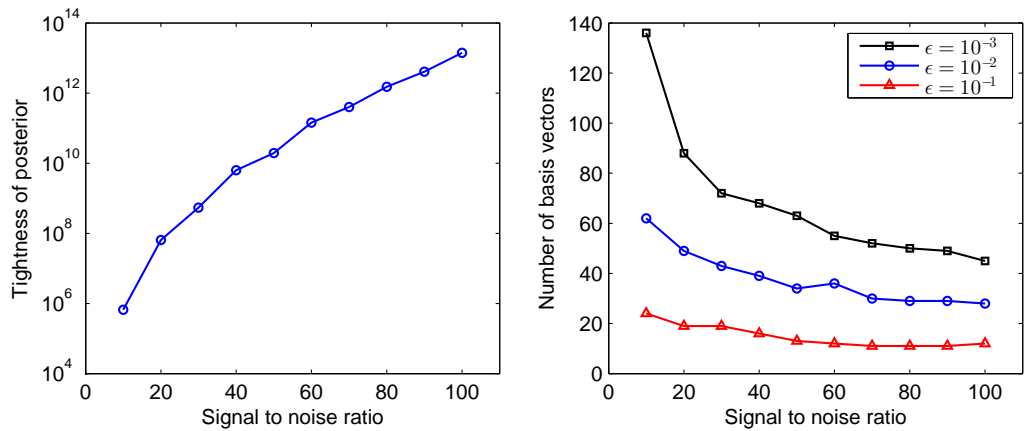}
\caption{Left: The posterior concentration (\ref{eq:tightness}) versus the signal-to-noise ratio of the data. Right: The number of reduced basis vectors versus the signal-to-noise ratio.}
\label{fig:num_basis}
\end{figure}

The amount of information carried in the data affects the dimension of the data-driven reduced-order model, and hence has an impact on its computational efficiency.
By adjusting the signal-to-noise ratio in the observed data, we examine the influence of the posterior concentration on the dimension of the reduced basis.
We gradually increase the signal-to-noise ratio from 10 to 100, and record the number of reduced basis vectors in the reduced-order models.
To quantify the posterior concentration, we use the ``tightness'' of the posterior distribution defined by
\begin{equation}
\prod_{i = 1}^{N_{\rm p}} \frac{\sigma_0(x_i)}{\sigma(x_i)},
\label{eq:tightness}
\end{equation}
where $\sigma(x_i)$ is the standard deviation of the posterior marginal of $x_i$, and $\sigma_0(x_i)$ is the standard deviation of the corresponding prior marginal. In Figure~\ref{fig:num_basis}, we observe that the dimension of the reduced basis decreases as the signal-to-noise ratio increases.
For this test problem, the larger amount of information in the data results in a lower dimensional reduced basis because our approach exploits the increasing concentration of the posterior.

\subsection{The high dimensional inverse problem}
\label{sec:hd_log}
\begin{figure}[h]
\centering
\includegraphics[width=\textwidth]{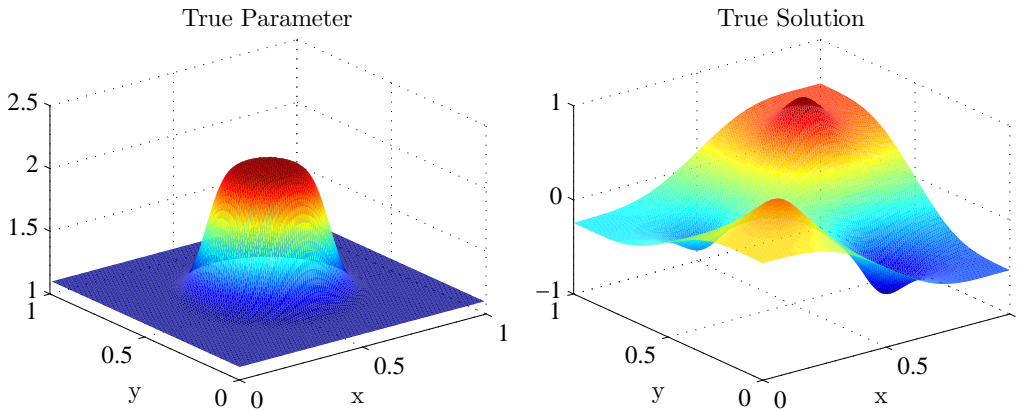}
\caption{Setup of the test case for the high-dimensional example. Left: The true permeability used for generating the synthetic data sets. Right: The model outputs of the true permeability. }
\label{fig:setup_hD_log}
\end{figure}
\begin{figure}[h]
\centering
\includegraphics[width=\textwidth]{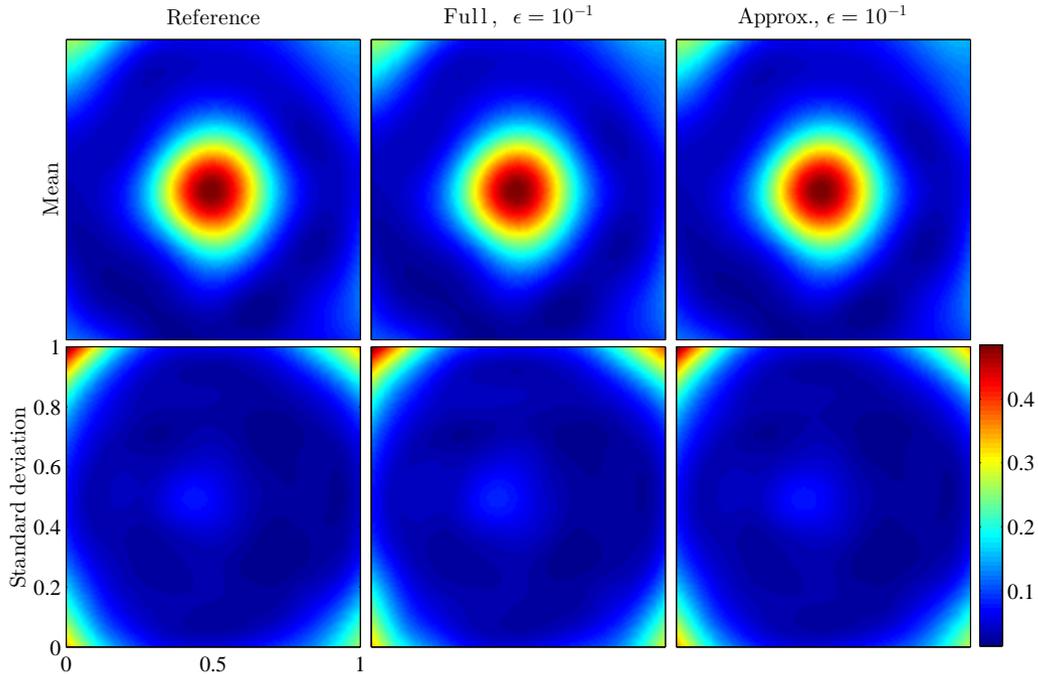}
\caption{Mean (top row) and standard deviation (bottom row) at each spatial location of the permeability field. From left to right: The reference algorithm, the full target algorithm with $\epsilon = 10^{-1}$, and the $\epsilon$-approximate algorithm  with $\epsilon = 10^{-1}$. }
\label{fig:result_hD_log}
\end{figure}
\begin{table}[h]
\caption{Comparison of the computational efficiency of the full target algorithm, the $\epsilon$-approximate algorithm, and the reference algorithm for the high-dimensional problem. $\epsilon = 10^{-1}$ is used in the full target algorithm. For the $\epsilon$-approximate algorithm, three examples with $\epsilon = \{10^{-1}, 10^{-2}, 10^{-3}\}$ are given. The posterior measure of the complement of the $\epsilon$-feasible set, $\mu(\Omega^{(m)}_\perp(\epsilon))$, is given in Definition \ref{def:feasible}.}
\label{table:performance_h_log}
\centering
\tabsize
\begin{tabular}{L{3cm}L{1.5cm}L{1.5cm}L{1.5cm}L{1.5cm}L{1.5cm}} \toprule
& \lower.3ex\hbox{Reference} & \lower.3ex\hbox{Full target} & \multicolumn{3}{c}{\lower.3ex\hbox{$\epsilon$-approximate}}  \\ \midrule 
\lower.3ex\hbox{Error threshold $\epsilon$} & - & \lower.3ex\hbox{$10^{-1}$} & \lower.3ex\hbox{$10^{-1}$} & \lower.3ex\hbox{$10^{-2}$} & \lower.3ex\hbox{$10^{-3}$} \\ \lower.3ex\hbox{full model evaluations} & \lower.3ex\hbox{$5 \times 10^5$} & \lower.3ex\hbox{$10^4$} & \lower.3ex\hbox{62} & \lower.3ex\hbox{$129$} & \lower.3ex\hbox{$209$} \\ 
\lower.3ex\hbox{Reduced basis vectors} & -  & \lower.3ex\hbox{$64$} & \lower.3ex\hbox{62} & \lower.3ex\hbox{$129$} & \lower.3ex\hbox{$209$} \\ 
\lower.3ex\hbox{CPU time (sec)} & \lower.3ex\hbox{$75300$}  & \lower.3ex\hbox{$1011$} & \lower.3ex\hbox{302} & \lower.3ex\hbox{$660$} & \lower.3ex\hbox{$1226$} \\ 
\lower.3ex\hbox{ESS} & \lower.3ex\hbox{$2472$} & \lower.3ex\hbox{$2221$} & \lower.3ex\hbox{$2468$} & \lower.3ex\hbox{$2410$} & \lower.3ex\hbox{2445}  \\ 
\lower.3ex\hbox{ESS / CPU time}  & \lower.3ex\hbox{$3.3\times10^{-2}$} & \lower.3ex\hbox{$ 2.2$} & \lower.3ex\hbox{$8.2$} & \lower.3ex\hbox{$3.7$} & \lower.3ex\hbox{$2.0$} \\ 
\lower.3ex\hbox{speedup factor} & \lower.3ex\hbox{$1$}  & \lower.3ex\hbox{$67$} & \lower.3ex\hbox{$249$} & \lower.3ex\hbox{$111$} & \lower.3ex\hbox{$61$}  \\ 
\lower.3ex\hbox{$\mu(\Omega^{(m)}_\perp(\epsilon))$} & - & \lower.3ex\hbox{$0$} & \lower.3ex\hbox{$0$} & \lower.3ex\hbox{$1.0 \times 10^{-3}$} & \lower.3ex\hbox{$1.6 \times 10^{-4}$} \\ \bottomrule
\end{tabular}
\end{table}
In the high-dimensional example, a log-normal distribution is employed to model the permeabilities as a random field.
Let $r_i, i = 1,\ldots,N_{\rm g}$ denote the coordinates of the $N_{\rm g}$ grid points.
Let $k(r_i) = \exp(x(r_i))$ be the permeability field defined at each grid point, then the latent field $\vecx = [x(r_1), \ldots, x(r_{N_{\rm g}}) ]^T$ follows a Gaussian process prior
\begin{equation}
\pi_0(\vecx) \propto \exp \left ( -\frac{1}{2} \vecx^T \Sigma \vecx \right), \;\;
\Sigma_{ij} = \exp \left( -\frac{| r_i - r_j |^2}{2 s^2}  \right),
\end{equation}
where $s = 0.25$ is used to provide sufficient spatial variability.
After applying the eigendecomposition of the prior covariance,
the parameters are defined on $43$ eigenvectors that preserve $99.99\%$ energy of the prior distribution.
To avoid an inverse crime, we use a ``true'' permeability field that is not directly drawn from the prior distribution. Figure \ref{fig:setup_hD_log} shows the true permeability field and the simulated pressure head. The setup of the measurement sensors is the same as the 9D example in Section~\ref{sec:9d}.

Using the same setting as the 9D case, we simulate the full target algorithm with $\epsilon = 10^{-1}$  for $10^4$ iterations, with subchain length $L = 50$. 
For these full target MCMC simulations, the first $2000$ samples are discarded as burn-in samples.
We simulate the $\epsilon$-approximate algorithm with $\epsilon = \{10^{-1}, 10^{-2}, 10^{-3}\}$, for $5\times 10^5$ iterations.
The single-stage MCMC method is simulated for $5\times 10^5$ iterations as the reference.
For all the $\epsilon$-approximate MCMC simulations and the reference MCMC simulation, the first $10^5$ samples are discarded as burn-in samples. 

Table \ref{table:performance_h_log} summarizes the number of full model evaluations, the number of reduced basis vectors, the CPU time, ESS, and speedup factor.
The speedup factor of the full target algorithm is about 67.
In comparison, the speedup factors of the $\epsilon$-approximate algorithm range from $61$ to $249$.
The speedup factor increases as the error threshold $\epsilon$ increases.
Figure~\ref{fig:result_hD_log} shows the mean and standard deviation at each spatial location of the permeability field, generated from the reference algorithm, the full target algorithm, and the least accurate setting ($\epsilon = 10^{-1}$) of the $\epsilon$-approximate algorithm.
We observe that all algorithms produce similar estimates of mean and standard deviation in this test case.

The $\mu(\Omega^{(m)}_\perp(\epsilon))$ values estimated from samples generated by the reference algorithm for all three $\epsilon$ values are also recorded in Table~\ref{table:performance_h_log}.
In this test example, we have $\mu(\Omega^{(m)}_\perp(\epsilon)) < \epsilon$ for all three $\epsilon$ values, and thus the Monte Carlo estimator provided by the $\epsilon$-approximate algorithm can be 
characterized by the $\epsilon$ values. 
We note that some of the estimated posterior measures $\mu(\Omega^{(m)}_\perp(\epsilon))$ have zero values in Table \ref{table:performance_h_log},  but these values do not necessarily mean that the posterior measures $\mu(\Omega^{(m)}_\perp(\epsilon))$ are exactly zero, because these are Monte Carlo estimates.

\subsection{Remarks on the $\epsilon$-approximate algorithm}
\label{sec:end_nr}
In the first case study, the $\epsilon$-approximate algorithm offers some speedup compared to the full target algorithm (range from $4.7$ to $7.4$).
In the second case study, the speedup factor of the $\epsilon$-approximate algorithm compared to the full target algorithm drops to at most $3.7$ (with $\epsilon = 10^{-1}$) and it performs slightly worse than the full target algorithm for $\epsilon = 10^{-3}$.
The speedup factor of the $\epsilon$-approximate algorithm decreases with $\epsilon$ in both cases.
This result is to be expected, since the computational cost of the reduced-order model depends on the the dimensionality of the reduced basis, which grows as $\epsilon$ decreases.
For $\epsilon = 10^{-3}$ in the second test case, the reduced-order model becomes computationally too expensive relative to the full model, and thus we lose the efficiency gain of the $\epsilon$-approximate algorithm over the full target algorithm.

For problems that require only a limited number of effective samples, using the $\epsilon$-approximate algorithm can be more advantageous.
This is because we can use a relatively large $\epsilon$ value to keep the computational cost of the reduced-order model low, while the resulting MSE is still dominated by the variance of the estimator rather than the bias.
If the goal is to obtain an accurate Monte Carlo estimator, in which the variance of the estimator is small compared to the bias resulting from sampling the approximate posterior distribution, we should use the full target algorithm.
We also note that the accuracy and efficiency of the $\epsilon$-approximate algorithm depend on reliable error indicators or error estimators, while the full target algorithm always samples the full posterior distribution, regardless of the error indicator.

\section{Conclusion}
We have introduced a new data-driven model reduction approach for solving statistical inverse problems. Our approach constructs the reduced-order model using adaptively-selected posterior samples to compute the snapshots. The reduced-order model construction process is integrated into the posterior sampling, to achieve simultaneous posterior exploration and model reduction.

Based on the data-driven reduced-order model, we have also developed two MCMC algorithms to sample the posterior distribution more efficiently than standard full-model MCMC algorithms. The full target algorithm aims to accelerate sampling of the full posterior distribution by coupling the full and approximate posterior distributions together. The $\epsilon$-approximate algorithm samples an approximate posterior distribution, and attempts to reduce the mean squared error of the resulting Monte Carlo estimator, compared to a standard MCMC algorithm.
Both algorithms adaptively construct the reduced-order model online through the MCMC sampling. The full target algorithm preserves ergodicity with respect to the true posterior. The $\epsilon$-approximate algorithm does not sample the full posterior, but can provide further speedups for some problems.

In the case studies, we have demonstrated that both algorithms are able to accelerate MCMC sampling of computationally expensive posterior distributions by up to two orders of magnitude, and that the sampling accuracy of the $\epsilon$-approximate algorithm is comparable to that of a reference full-model MCMC. We have also used the first case study to show the numerical accuracy of the data-driven reduced-order model, compared to a reduced-order model that is built offline with respect to the prior distribution. In this example, for the same number of reduced basis vectors, the posterior-averaged output error of the data-driven reduced-order model is several orders of magnitude smaller than that of the reduced-order model built with respect to the prior. Furthermore, we have demonstrated the impact of the amount of information carried in the observed data on the dimensionality of the reduced basis.

For solving statistical inverse problems, these results suggest that a data-driven reduced-order model is preferable to a reduced-order model built with respect to the prior, especially when the data are informative. Even though our approach is designed for constructing projection-based reduced-order models, the concept of building posterior-oriented surrogate models can be generalized to other approximation approaches such as Gaussian process regression and generalized polynomial chaos.

\acks
The authors thank Florian Augustin, Tan Bui-Thanh, Omar Ghattas, and Jinglai Li for many helpful comments and discussions.
This work was supported by the United States Department of Energy Applied Mathematics Program, Awards DE-FG02-08ER2585 and DE-SC0009297, as part of the DiaMonD Multifaceted Mathematics Integrated Capability Center.


\begin{thebibliography}{10}
\providecommand{\url}[1]{\texttt{#1}}
\providecommand{\urlprefix}{URL }
\expandafter\ifx\csname urlstyle\endcsname\relax
  \providecommand{\doi}[1]{doi:\discretionary{}{}{}#1}\else
  \providecommand{\doi}{doi:\discretionary{}{}{}\begingroup
  \urlstyle{rm}\Url}\fi

\bibitem{Kaipio_2005}
Kaipio JP, Somersalo E. \emph{Statistical and computational inverse problems},
  vol. 160. Springer: New York, 2004.

\bibitem{Tarantola_2004}
Tarantola A. \emph{Inverse Problem Theory and Methods for Model Parameter
  Estimation}. Society for Industrial Mathematics: Philadelphia, 2005.

\bibitem{Liu_2001}
Liu JS. \emph{{M}onte {C}arlo strategies in scientific computing}. Springer:
  New York, 2001.

\bibitem{CFO_2011}
Cui T, Fox C, O'Sullivan MJ. {B}ayesian calibration of a large-scale geothermal
  reservoir model by a new adaptive delayed acceptance {M}etropolis {H}astings
  algorithm. \emph{Water Resource Research}  2011; \textbf{47}:W10\,521.

\bibitem{HLH_2003}
Higdon D, Lee H, Holloman C. {M}arkov chain monte carlo-based approaches for
  inference in computationally intensive inverse problems. \emph{{B}ayesian
  Statistics 7}, Bernardo JM, Bayarri MJ, Berger JO, Dawid AP, Heckerman D,
  Smith AFM, West M (eds.). Oxford University Press, Oxford 2003; 181--197.

\bibitem{Mckeague_2005}
McKeague IW, Nicholls GK, Speer K, Herbei R. Statistical inversion of south
  atlantic circulation in an abyssal neutral density layer. \emph{Journal of
  Marine Research}  2005; \textbf{63}:683--704.

\bibitem{Haario_2004}
Haario H, Laine M, Lehtinen M, Saksman E, Tamminen J. {M}arkov chain {M}onte
  {C}arlo methods for high dimensional inversion in remote sensing.
  \emph{Journal of the Royal Statistical Society: Series B (Statistical
  Methodology)}  2004; \textbf{66}:591--608.

\bibitem{Martin_2012}
Martin J, Wilcox LC, Burstedde C, Ghattas O. A stochastic newton {MCMC} method
  for large-scale statistical inverse problems with application to seismic
  inversion. \emph{SIAM Journal on Scientific Computing}  2012;
  \textbf{34}(3):A1460--A1487.

\bibitem{Marzouk_2007}
Marzouk YM, Najm HN, Rahn LA. Stochastic spectral methods for efficient
  bayesian solution of inverse problems. \emph{Journal of Computational
  Physics}  2007; \textbf{224}:560--586.

\bibitem{Marzouk_2009}
Marzouk YM, Najm HN. Dimensionality reduction and polynomial chaos acceleration
  of {B}ayesian inference in inverse problems. \emph{Journal of Computational
  Physics}  2009; \textbf{228}:1862--1902.

\bibitem{Bayarri_2009}
Bayarri MJ, Berger J, Kennedy M, Kottas A, Paulo R, Sacks J, Cafeo J, Lin C, Tu
  J. Predicting vehicle crashworthiness: Validation of computer models for
  functional and hierarchical data. \emph{Journal of the American Statistical
  Association}  2009; \textbf{104}:929--943.

\bibitem{Galbally_2008}
Galbally D, Fidkowski K, Willcox KE, Ghattas O. Nonlinear model reduction for
  uncertainty quantification in large scale inverse problems.
  \emph{International journal for numerical methods in engineering}  2008;
  \textbf{81}(12):1581--1608.

\bibitem{Lipponen_2013}
Lipponen A, Sepp\"{a}nen A, Kaipio JP. Electrical impedance tomography imaging
  with reduced-order model based on proper orthogonal decomposition.
  \emph{Journal of Electronic Imaging}  2013; \textbf{22}:023\,008.

\bibitem{Lieberman_2010}
Lieberman C, Willcox KE, Ghattas O. Parameter and state model reduction for
  large-scale statistical inverse problems. \emph{SIAM Journal on Scientific
  Computing}  2010; \textbf{32}(5):2523--2542.

\bibitem{Wang_2005}
Wang J, Zabaras N. Using bayesian statistics in the estimation of heat source
  in radiation. \emph{International Journal of Heat and Mass Transfer}  2005;
  \textbf{48}:15--29.

\bibitem{Sirovich_1987}
Sirovich L. Turbulence and the dynamics of coherent structures. part 1:
  Coherent structures. \emph{Quarterly of Applied Mathematics}  1987;
  \textbf{45}:561--571.

\bibitem{Cover_2012}
Cover TM, Thomas JA. \emph{Elements of information theory}. Wiley-interscience, New York
  2006.

\bibitem{Arian_2000}
Arian E, Fahl M, Sachs EW. Trust-region proper orthogonal decomposition for
  flow control. \emph{Technical {R}eport ICASE-2000-25}, Institute for computer
  applications in science and engineering, Hampton, VA 2000.

\bibitem{Ravindran_2002}
Ravindran SS. Adaptive reduced-order controllers for a thermal flow system
  using proper orthogonal decomposition. \emph{SIAM Journal on Scientific
  Computing}  2002; \textbf{23(6)}:1924--1942.

\bibitem{KV_2008}
Kunisch K, Volkwein S. Proper orthogonal decomposition for optimality systems.
  \emph{ESAIM: Mathematical Modelling and Numerical Analysis}  2008;
  \textbf{42(1)}:1--23.

\bibitem{Carlberg_2008}
Carlberg K, Farhat C. A compact proper orthogonal decomposition basis for
  optimization-oriented reduced-order models. \emph{AIAA Paper}  2008;
  \textbf{5964}:10--12.

\bibitem{Stuart_2010}
Stuart AM. Inverse problems: a {B}ayesian perspective. \emph{Acta Numerica}
  2010; \textbf{19}:451--559.

\bibitem{Grenander_1994}
Grenander U, Miller MI. Representations of knowledge in complex systems.
  \emph{Journal of the Royal Statistical Society: Series B (Statistical
  Methodology)}  1994; \textbf{56}(4):549--603.

\bibitem{Hastings}
Hastings W. {M}onte {C}arlo sampling using {M}arkov chains and their
  applications. \emph{Biometrika}  1970; \textbf{57}:97--109.

\bibitem{Metropolis}
Metropolis N, Rosenbluth AW, Rosenbluth MN, Teller AH, Teller E. Equation of
  state calculations by fast computing machines. \emph{Journal of chemical
  physics}  1953; \textbf{21}:1087--1092.

\bibitem{Haario_2001}
Haario H, Saksman E, Tamminen J. An adaptive {M}etropolis algorithm.
  \emph{Bernoulli}  2001; \textbf{7}(2):223--242.

\bibitem{Roberts_2007}
Roberts GO, Rosenthal JS. Coupling and ergodicity of adaptive {M}arkov chain
  {M}onte {C}arlo algorithms. \emph{Journal of applied probability}  2007;
  \textbf{44}(2):458--475.

\bibitem{Girolami_2011}
Girolami M, Calderhead B. Riemann manifold {L}angevin and {H}amiltonian {M}onte
  {C}arlo methods. \emph{Journal of the Royal Statistical Society: Series B
  (Statistical Methodology)}  2011; \textbf{73}(2):123--214.

\bibitem{Liu_1998}
Liu JS, Chen R. Sequential {M}onte {C}arlo methods for dynamic systems.
  \emph{Journal of the American statistical association}  1998;
  \textbf{93}(443):1032--1044.

\bibitem{CF_2005}
Christen JA, Fox C. {MCMC} using an approximation. \emph{Journal of
  Computational and Graphical statistics}  2005; \textbf{14}(4):795--810.

\bibitem{Astrid_2008}
Astrid P, Weiland S, Willcox KE, Backx T. Missing point estimation in models
  described by proper orthogonal decomposition. \emph{IEEE Transactions on
  Automatic Control}  2008; \textbf{53}(10):2237--2251.

\bibitem{BMNP_2004}
Barrault M, Maday Y, Nguyen NC, Patera AT. An ``empirical interpolation''
  method: application to efficient reduced-basis discretization of partial
  differential equations. \emph{Comptes Rendus Mathematique}  2004;
  \textbf{339}(9):667--672.

\bibitem{Chat_2010}
Chaturantabut S, Sorensen DC. Nonlinear model reduction via discrete empirical
  interpolation. \emph{SIAM Journal on Scientific Computing}  2010;
  \textbf{32}(5):2737--2764.

\bibitem{HDO_2011}
Haasdonk B, Dihlmann M, , Ohlberger M. A training set and multiple bases
  generation approach for parameterized model reduction based on adaptive grids
  in parameter space. \emph{Math. Comput. Model. Dyn. Syst.}  2011;
  \textbf{17}:423--442.

\bibitem{AZF_2012}
Amsallem D, Zahr M, Farhat C. Nonlinear model order reduction based on local
  reduced-order bases. \emph{International Journal Numerical Methods
  Engineering}  2012; \textbf{92}:891--916.

\bibitem{ES_2012}
Eftang JL, Stamm B. Parameter multi?domain `hp' empirical interpolation.
  \emph{International Journal for Numerical Methods in Engineering}  2012;
  \textbf{90}(4):412--428.

\bibitem{PBWB_2014}
Peherstorfer B, Butnaru D, Willcox KE, Bungartz HJ. Localized discrete
  empirical interpolation method. \emph{SIAM Journal on Scientific Computing}
  36; \textbf{2014}(1):A168--A192.

\bibitem{Meyer_2003}
Meyer M, Matthies HG. Efficient model reduction in non-linear dynamics using
  the karhunen-loeve expansion and dual-weighted-residual methods.
  \emph{Computational Mechanics}  2003; \textbf{31}(1):179--191.

\bibitem{Cui_2010}
Cui T. Bayesian calibration of geothermal reservoir models via markov chain
  monte carlo. Ph{D} {T}hesis, The University of Auckland, Auckland 2010.

\bibitem{Bui_2008_2}
Bui-Thanh T, Willcox KE, Ghattas O. Model reduction for large-scale systems
  with high-dimensional parametric input space. \emph{SIAM Journal on
  Scientific Computing}  2008; \textbf{30}(6):3270--3288.

\bibitem{Patera_2007}
Patera AT, Rozza G. \emph{Reduced basis approximation and a posteriori error
  estimation for parametrized partial differential equations}. MIT Pappalardo
  monographs in mechanical engineering (to appear), Copyright MIT (2006--2007),
  2007.

\end{thebibliography}
\end{document}